\documentclass[journal]{IEEEtran}
\usepackage{enumerate}
\usepackage{enumitem}

\usepackage[innermargin=0.75in,outermargin=0.75in,top=0.68in,bottom=1.48in]{geometry}
\usepackage{preamble}
\usepackage{algorithm,algorithmic}
\usepackage{float}
\newfloat{algorithm}{t}{lop}
\PassOptionsToPackage{bookmarks=false}{hyperref}
\usepackage{bbm}
\usepackage{lipsum,multicol}
\usepackage{booktabs}

\usepackage{tikz,tikz-cd}
\usepackage{multicol}
\usepackage{blkarray}
\usetikzlibrary{shapes.geometric,decorations.pathreplacing,calc,scopes}
\tikzstyle{r} = [draw, very thick, red,-]
\tikzstyle{g} = [draw, very thick, green, -]
\tikzstyle{d} = [draw, very thick,black, dashed]
\tikzstyle{b} = [draw, very thick, black, -]
\tikzstyle{a} = [draw, very thick, black, ->]
\tikzstyle{la} = [draw, thick, black, ->]
\tikzstyle{t} = [draw,thin,black,-]
\tikzstyle{blueline} = [draw, very thick, blue, -]
\tikzstyle{backa} = [draw, very thick, black, <-]

\def\ppsmatrix#1{\begin{psmallmatrix}#1\end{psmallmatrix}} 

\usepackage{mathrsfs}
\usepackage{tkz-euclide,subfigure}
\usetikzlibrary{shapes.geometric, arrows}
\usepackage{caption}
\newcommand{\sfX}{{\sf X}}

\newcommand{\sfZ}{{\sf Z}}

\begin{document}
\title{Absorbing Sets in Quantum LDPC Codes}
\author{
  \IEEEauthorblockN{Kirsten D. Morris, Tefjol Pllaha, Christine A. Kelley}\\
    	\IEEEauthorblockA{\small University of Nebraska-Lincoln, Lincoln, NE, USA.\\ 
    	E-mail: kmorris11@huskers.unl.edu, tefjol.pllaha@unl.edu, ckelley2@unl.edu
	}
\thanks{Part of this work was presented at the 2023 International Symposium on Topics in Coding (ISTC)~\cite{10273573}.
K. M. thanks the support of the GFSD and Othmer fellowships. This work was also supported by a Simons Collaboration Grant.}
}

\IEEEoverridecommandlockouts

\maketitle

\begin{abstract}
Iterative decoder failures of quantum low density parity check (QLDPC) codes are attributed to substructures in the code’s graph, known as trapping sets, as well as degenerate errors that can arise in quantum codes. 
Failure inducing sets are subsets of codeword coordinates that, when initially in error, lead to decoding failure in a trapping set.
The purpose of this paper is to examine failure inducing sets of QLDPC codes under syndrome-based iterative decoding.
As for classical LDPC codes, we show that absorbing sets play a central role in understanding decoder failures.
Raveendran and Vasić~\cite{raveendran2021trapping} initiated the study of quantum trapping sets, where beyond the classical-type trapping sets, they identified rigid symmetric structures (a.k.a symmetric stabilizers) responsible for degenerate errors. 
In this paper, we show that this behavior is part of a much more general phenomenon that can be described by the absorbing set framework.

\end{abstract}

\section{Introduction}

Quantum information is far from perfect and very much prone to errors. For this reason, fault-tolerant quantum computation is a must and quantum error correction becomes a central topic.
Since the discovery of the first quantum error correction code~\cite{Shor-first}, there has been tremendous progress in code design.
Many of the existing quantum codes leverage the vast existing literature in classical coding theory. Low density parity check (LDPC) codes are well-established in classical coding theory. 
Their quantum analogues, quantum LDPC (QLDPC) codes, gained popularity due to Gottesman's breakthrough discovery~\cite{Gottesman-overhead}, showing that constant overhead can be achieved with constant encoding rate.
For this reason, QLDPC codes are current candidates for realizing scalable fault-tolerant quantum computation.  
Like their classical counterparts, QLDPC codes are amenable to low complexity iterative decoding algorithms, such as syndrome-based iterative decoding  \cite{RBV19, RRPV22}.
In practice, these algorithms are run until either an estimated error pattern is obtained or a maximum number of iterations is reached. However, these algorithms are suboptimal on finite length QLDPC codes, meaning that they do not always produce the correct estimated error pattern. Instead, the decoder may fail to converge or output an erroneous syndrome, or it may estimate an error pattern with the same syndrome that differs from the actual pattern by a logical operator.

Iterative decoders may be viewed as graph-based algorithms that operate on the code's Tanner graph, which is the graphical representation of the code's parity check matrix. Failure of iterative decoding of LDPC codes has been attributed to graphical substructures, called {\em trapping sets}, in the Tanner graph \cite{R03}.  These structures contribute to persistent error floors in the Bit Error Rate (BER) or Frame Error Rate (FER) curves of these codes.  Moreover, these structures naturally depend on the choice of Tanner graph representation used in the decoding process. Failure inducing sets of a trapping set are subsets of the codeword coordinates that, when initially in error, lead to  a decoding failure. While most work on decoder failure of QLDPC codes has focused on defining and identifying trapping sets of QLDPC codes \cite{raveendran2021trapping}, less has characterized  failure inducing sets of trapping sets.  This work aims to identify classes of failure inducing sets, and takes initial steps at predicting the type of error that results.

Iterative decoder failure of classical LDPC codes on different channels is  attributed to graphical substructures, such as stopping sets, trapping sets, and absorbing sets \cite{Di02, R03, dolecek07}. 
We observed that many failure inducing sets of QLDPC trapping sets were in fact absorbing sets. 
Since absorbing sets and trapping sets are closely related in structure, it is natural to explore the role of absorbing sets in syndrome-based iterative decoder failures of QLDPC codes. 
In this paper we examine the connection between absorbing sets, as they are defined for classical LDPC codes, and trapping sets, and identify cases when absorbing sets are trapping sets and failure inducing sets.

Degenerate errors are unique to quantum error correction. 
Since such errors have the same syndrome, a syndrome-based decoder might attempt to converge at any such error thereby oscillating from one to another and ultimately diverging.
In~\cite{raveendran2021trapping}, such behavior was attributed to certain graphical structures, called symmetric stabilizer; see Definition~\ref{D-ss}.
The main structural features of a symmetric stabilizer are that (1) it can be ``partitioned" as an even number of isomorphic substructures and (2) said substructures have identical sets of odd degree check nodes.
In a series of results and examples, we show that neither of the restrictions are necessary.
Moreover, we show that the framework of~\cite{raveendran2021trapping} can be explained in terms of absorbing sets.

This paper is organized as follows. In Section II, we introduce the necessary notation and background on quantum stabilizer codes, QLDPC codes and their graph representation, and syndrome-based iterative decoding.  
In Section III we examine graph structures that affect decoder performance. 
In Section IV we analyze absorbing sets and identify cases when absorbing sets are  trapping sets with respect to the syndrome decoder, and when they are failure inducing sets. 
In Section V we extend these results to absorbing sets embedded within larger graphs and crystallize a general framework that incorporates the symmetric stabilizer framework.
This section also includes a case study using hypergraph-product codes. 
We conclude the paper in Section VI with remarks for future work.
\section{Preliminaries}

\subsection{Stabilizer formalism} 
Stabilizer codes~\cite{Gottesman-phd97} are quantum codes obtained as the simultaneous eigenspace of commuting Pauli matrices. Specifically, the Pauli group acting on one physical qubit $\C^2$, denoted $\cP_1$, is generated by ${\sf I}_2 = \ppsmatrix{1&0\\0&1}, {\sf X} = \ppsmatrix{0&1\\1&0}, {\sf Z} = \ppsmatrix{1&0\\0&-1}$, and ${\sf Y} = i{\sf XZ}$. 
It is convenient to think of the Pauli matrices in terms of the binary representation ${\sf I}_2 \to (0,0)$, ${\sf X} \to (1,0)$, ${\sf Z} \to (0,1)$, and ${\sf Y} \to (1,1)$. 
The Pauli group acting on $n$ physical qubits $\C^{\otimes n}\cong\C^{2^n}$, denoted $\cP_n$, is then naturally the $n$-fold Kronecker product of $\cP_1$.
A stabilizer group $\cS\leq \cP_n$ is an abelian group that does not contain $-{\sf I}_{2^n}$. 
A stabilizer with $k$ independent generators can be naturally represented as an $k\times 2n$ matrix $H_\cS = \ppsmatrix{H_{\sf X}&\mid&H_{\sf Z}}$, and it defines an $[\![n,n-k]\!]$ quantum code. 
Two stabilizer generators commute if and only if their respective representations $h = (h_{\sf X}, h_{\sf Z}), g = (g_{\sf X},g_{\sf Z})$ of $H$ are orthogonal with respect to the {\em symplectic inner product} $h\odot g = h_{\sf X}g_{\sf Z}^T + h_{\sf Z}g_{\sf X}^T$.
Cumulatively, this leads to
\begin{equation}\label{e-sip}
H_\cS\odot H_\cS:=H_{\sf X}H_{\sf Z}^T + H_{\sf Z}H_{\sf X}^T = 0.
\end{equation}
The {\em logical operators} acting on the code space correspond to the elements of the Pauli group that commute with $\cS$. Elements of $\cS$ commute with each other so they are naturally (trivial) logical operators.

An important class of stabilizer codes are the Calderbank-Shor-Steane (CSS) codes~\cite{CS96,Steane96}, defined by a pair of classical linear codes $C_{\sf X}, C_{\sf Z}\subset \Fq^n$ such that $C_{\sf X}^\perp\subset C_{\sf Z}$. 
This condition forces two respective parity check matrices $H_{\sf X}$ and $H_{\sf Z}$ to satisfy $H_{\sf Z}H_{\sf X}^T = 0$ and thus the matrix
\[
H = \left(\!\!\begin{array}{c|c}H_{\sf X}&0\\0&H_{\sf Z} \end{array}\!\!\right)
\]
satisfies~\eqref{e-sip} and it defines a stabilizer code. 

\subsection{QLDPC codes and syndrome-based iterative decoding}
Low density parity check (LDPC) codes are codes characterized by having sparse parity check matrix representations.
Given a parity check matrix $H$ of an LDPC code, its bipartite Tanner graph representation is the graph $\cG = (V,W; E)$
where the vertex sets $V$ and $W$ correspond to the codeword coordinates and the parity check equations, respectively, and $E$ is the set of edges. Vertices in $V$ and $W$ are called variable and check nodes, respectively. For $v_i \in V$ and $c_j \in W$, the edge $(v_i,c_j)\in E$ if and only if $h_{j,i} = 1$ in $H$. The sparsity of $H$ ensures that the graph is sparse, making it amenable to low complexity iterative decoders. Indeed, the complexity of iterative decoders is linear in the number of edges \cite{T81}.

Tanner graphs are defined similarly for quantum CSS codes. However, since the parity check equations can be partitioned into those that have nonzero entries in  $H_{\sf X}$ and those that have nonzero entries in $H_{\sf Z}$, the variable nodes have two edge types, those determined by ${\sf X}$ errors and those determined by ${\sf Z}$ errors, and each check node is incident to only one edge type.

Typically, for stabilizer codes and quantum codes in general there is a correlation between ${\sf X}$ and ${\sf Z}$ errors. However, for CSS codes we can ignore such correlation~\cite{mackay2004sparse} and treat them over two independent binary symmetric channels.
Let $e = (e_\sfX,e_\sfZ)$ be the binary representation of a Pauli error acting on $n$ qubits. The corresponding error syndrome captures the commutativity/orthogonality relations of the error with each of the stabilizers, that is,
\begin{align*}
\sigma_e & = (\sigma_\sfX, \sigma_\sfZ)
= \left(\!\!\begin{array}{c|c}H_{\sf X}&0\\0&H_{\sf Z} \end{array}\!\!\right)\odot e 
= (H_\sfZ e_\sfX^T, H_\sfX e_\sfZ^T).
\end{align*}
Thus, $H_{\sf X}$ can be used to decode ${\sf Z}$ errors and $H_{\sf X}$ can be used to decode ${\sf Z}$ errors.
An all-zero syndrome indicates that the error $e$ commutes with all the stabilizers and thus it is undetectable. 
If $f$ is itself a stabilizer, that is, it belongs to the rowspace of $H$, then $\sigma_f = \vec{0}$. It follows that for any Pauli error $e$ we have $\sigma_{e+f} = \sigma_e$. This means that decoding can be only performed up to stabilizers. Two Pauli errors $e, f$ are called {\em degenerate errors} if they yield the same syndrome $\sigma_e = \sigma_f$, or equivalently, if $e+f$ is a stabilizer. Such errors have no classical analog.

The goal of a syndrome-based decoder is to match the input syndrome. Specifically, the decoder outputs estimated errors $\hat{e}$ whose syndrome $\sigma_{\hat{e}}$ matches the input syndrome $\sigma_e$.
Once the syndromes are matched, the estimated error is applied to correct the error introduced by the channel.
Error correction fails if the decoder fails to match the syndrome or if there is a mis-correction, that is, the decoder produces a logical error. In particular, a logical error occurs if $e+\hat{e}$ is not a stabilizer.

To summarize, the possible outcomes of syndrome decoding are the following. 
If the estimated syndrome $\sigma_{\hat{e}}$ matches the input syndrome $\sigma_e$ and $\hat{e}=e$, the decoder recovered the exact error pattern. If $\sigma_{\hat{e}}$ matches $\sigma_e$ and $e +\hat{e}$ is in the rowspace of $H$, the decoder recovered a degenerate error $\hat{e}$, which is still considered successful decoding. Decoding failure occurs if there is a logical error, meaning $\sigma_{\hat{e}}$ matches $\sigma_e$ but $e + \hat{e}$ is not a stabilizer. Decoding failure also occurs if the estimated syndrome $\sigma_{\hat{e}}$ never matches the input syndrome $\sigma_e$. This can occur either when the estimated syndrome oscillates and never converges to the correct syndrome, or if the estimated syndrome converges to $\sigma_{\hat{e}}$ which does not match $\sigma_e$.

We now present the Gallager-B Syndrome-based Iterative Decoding Algorithm over the Binary Symmetric Channel (BSC). Let $e=(e_1, e_2, \dots, e_n) \in \mathbb{F}_2^n$ denote an error pattern, where $e_i = 1$ if variable node $v_i$ is in error, and $e_i = 0$ otherwise. Thus, $e$ is an incidence vector of the error locations. The neighboring checks of the variable nodes are either satisfied or unsatisfied, where an unsatisfied check means the incoming messages sum to $1\pmod 2$ and a satisfied check means the incoming messages sum to $0\pmod 2$. These check node values correspond to the input syndrome $\sigma=(\sigma_1, \sigma_2, \dots, \sigma_k)$.

For decoding, an all-zero error pattern is initially assumed. That is, all outgoing message symbols $\hat{e}_i$ are set to $0$. The outgoing check node message over an edge is computed as the XOR of extrinsic variable node messages and syndrome input value. The outgoing variable node message is the majority value among incoming extrinsic check node messages. If there is a tie, then the value of $0$ is sent, since a low weight error pattern is assumed.
The error pattern at the $\ell^{th}$ iteration, denoted $\hat{e}^{\ell}$, is determined to be the majority among \textit{all} incoming check node values at each variable node. If there is a tie, there is assumed to be no error.
The output syndrome value for the $i^{th}$ check node $c_i$ in the $\ell^{th}$ iteration is   
    $$\hat{\sigma}_i^{\ell} := \sum_{j \in \mathcal{N}(c_i)} \hat{e}_j \pmod 2$$
where the sum is taken over all incoming messages $\hat{e}_j$ in the neighborhood of the $c_i^{th}$ check node and is computed modulo 2. A check node $c_i$ is matched if and only if $\hat{\sigma}_i=\sigma_i$. If all syndrome values are matched, the iterative decoder outputs the error pattern $\hat{e}$. If not, the decoder repeats the previous steps.

\section{Iterative decoder failure and trapping sets}

In this section, we provide backgound on trapping sets, failure inducing sets, and absorbing sets, and show via examples how they affect iterative decoder performance. Moreover, we also illustrate the different types of decoder outcomes that can happen under syndrome-based iterative decoding.

Given a Tanner graph $\cG = (V,W; E)$ and a subset $S$ of $V$, let $\mathcal{N}(S)$ denote the set of check nodes that are incident to vertices in $S$. Let $\cG_S = (S,W_S;E_S)$, where $W_S = \mathcal{N}(S)$, denote the subgraph induced by $S \cup W_S$ in $\cG$. Thus,   $E_S$ is the set of edges in $\cG$ that have one vertex in $S$ and the other in $W_S$. Let $\mathcal{C}$ be a binary LDPC code of length $n$ with associated Tanner graph $\cG$, to be decoded with some chosen hard- or soft-decision decoder. Suppose that the codeword $\mathbf{x}$ is transmitted, and $\mathbf{y}$ is received. Let $\mathbf{y}^{\ell}$ be the output after $\ell$ iterations of the decoder are run on $\cG$, with  input syndrome $\sigma=(\sigma_1, \sigma_2, \dots, \sigma_k)$. Let $[n]$ denote the set $\{1, 2, \dots, n\}$. 

\begin{defi}
For $i \in [n]$, a variable node $v_{i}$ is {\em eventually correct} if there exists $L\in \mathbb{Z}_{\geq 0}$ such that $y_{i}^{\ell}=x_{i}$ for all $\ell\geq L$. Similarly, for $j \in [k]$, a check node $c_j$ is {\em eventually correct} if there exists $L\in \mathbb{Z}_{\geq 0}$ such that $\hat{\sigma_{j}}^{\ell}=\sigma_{j}$ for all $\ell\geq L$. 
\end{defi}

\begin{defi} For $i \in [n]$, a variable node $v_{i}$ {\em eventually converges} if there exists $L\in \mathbb{Z}_{\geq 0}$ such that $\hat{e}_{i}^{\ell}=\hat{e}_{i}^{(\ell + 1)}$ for all $\ell\geq L$.  Similarly, for $j \in [k]$, a check node $c_j$ {\em eventually converges} if there exists $L\in \mathbb{Z}_{\geq 0}$ such that $\hat{\sigma_{j}}^{\ell}=\hat{\sigma}_{j}^{(\ell+1)}$ for all $\ell\geq L$. 
\end{defi}

Note that if  variable node $v_i$ eventually converges, it may or may not converge to the correct estimate, $e_i$.

\begin{defi} A (quantum) trapping set  for a syndrome-based iterative decoder is a non-empty set  of variable nodes $\mathcal{T}$ in a Tanner graph $\cG$ such that there is a subset of variable nodes $\mathcal{F} \subseteq \mathcal{T}$ that when initially in error result in some subset of check nodes of $\mathcal{N}(\mathcal{T})$ that are not eventually correct and/or some subset of variable nodes of $\mathcal{T}$ that do not eventually converge. Such a subset $\mathcal{F}$ of variable nodes that when initially in error result in a trapping set $\mathcal{T}$ is called a {\em failure inducing} set for $\cT$. If the induced subgraph $\cG_\cT$ has $a$ variable nodes and $b$ odd degree check nodes, then $\mathcal{T}$ is said to be an {\em $(a,b)$-trapping set}. 
\end{defi}

Although $\cG_\cT$ is induced by $\cT \cup \cN(\cT)$, the graph $\cG_\cT$ is often referred to in the literature as a trapping set (TS) induced subgraph. To analyze decoder failure, we follow the convention of assuming messages outside the trapping set are correct.

\begin{defi}  
The {\em critical number} of a trapping set $\mathcal{T}$, denoted $\mu(\cT)$, is the smallest number of variable nodes in a failure inducing set for $\mathcal{T}$. The {\em strength} of a trapping set $\mathcal{T}$ is the number of failure inducing sets of cardinality $\mu$.
\end{defi}
\begin{figure}[!h]
\centering
\resizebox{0.48\textwidth}{!}{
    \centering
    \subfigure[\empty]{
         \begin{tikzpicture}[thick,scale=.8]
\node[circle,draw=black,label={above:\small{$v_1$}}] (v1) at (-2,2){};
\node[circle,draw=black,label={above:\small{$v_2$}}] (v2) at (2,2){};
\node[circle,draw=black,label={below:\small{$v_3$}}] (v3) at (2,-2){};
\node[circle,draw=black,label={below:\small{$v_4$}}] (v4) at (-2,-2){};

\node[shape=rectangle,draw=black,label={above:\small{$c_1$}}] (c1) at (0,2){};
\node[shape=rectangle,draw=black,label={right:\small{$c_2$}}] (c2) at (2,0){};
\node[shape=rectangle,draw=black,label={below:\small{$c_3$}}] (c3) at (0,-2){};
\node[shape=rectangle,draw=black,label={left:\small{$c_4$}}] (c4) at (-2,0){};
\node[shape=rectangle,draw=black,label={above:\small{$c_6$}}] (c6) at (0,0){};
\node[shape=rectangle,draw=black,label={below:\small{$c_5$}}] (c5) at (1.5,1.5){};
\node[shape=rectangle,draw=black,label={right:\small{$c_7$}}] (c7) at (-1.5,-1.5){};

\path [-,thick] (v1) edge node[left] {} (c1); \path [-,thick] (c1) edge node[left] {} (v2);
\path [-,thick] (v1) edge node[left] {} (c4); \path [-,thick] (c4) edge node[left] {} (v4);
\path [-,thick] (v4) edge node[left] {} (c3); \path [-,thick] (c3) edge node[left] {} (v3);
\path [-,thick] (v2) edge node[left] {} (c2); \path [-,thick] (c2) edge node[left] {} (v3);
\path [-,thick] (v1) edge node[left] {} (c6); \path [-,thick] (c6) edge node[left] {} (v3);
\path [-,thick] (v2) edge node[left] {} (c5);
\path [-,thick] (v4) edge node[left] {} (c7);

\end{tikzpicture}
         }
        \hfill
         \subfigure[\empty]
         {\raisebox{20mm}{
         \begin{tabular}{c|c|c}
    \toprule
    \textbf{Failure Inducing} & \textbf{VNs not} & \textbf{CNs not} \\
    \textbf{Set} & \textbf{Eventually Converged} & \textbf{Eventually Satisfied}\\
    \midrule
    $\{v_1, v_2, v_3, v_4\}$ & None & $\{c_5, c_{7}\}$ \\
    $ \{v_1, v_2, v_4\}$ & $\{v_1, v_2, v_4\}$ & $\{c_1, c_2, c_3, c_4, c_5, c_6, c_{7}\}$ \\
    $ \{v_2, v_3, v_4\}$ & $\{v_2, v_3, v_4\}$ & $\{c_1, c_2, c_3, c_4, c_5, c_6, c_{7}\}$ \\
    $\{v_1, v_3, v_4\}$ & None & $\{c_1, c_2, c_5, c_{7}\}$ \\
    $\{v_1, v_2, v_3\}$ & None & $\{c_3, c_4, c_5, c_{7}\}$ \\
    \bottomrule
  \end{tabular}
         }
         }
         }
\vspace{-.05 in}\caption{Left: A graph induced by a $(4,2)$-trapping set $\mathcal{T}$. Right: Failure inducing sets of $\mathcal{T}$.}
\label{fig:TSexample1}
\end{figure}
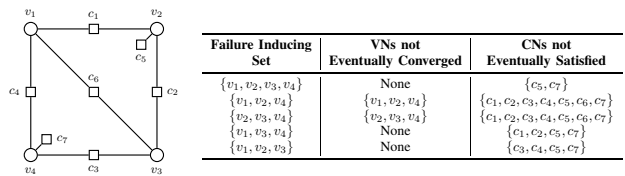

\begin{exa} Figure ~\ref{fig:TSexample1} shows a subgraph induced by a $(4,2)$-trapping set $\mathcal{T} = \{v_1,v_2,v_3,v_4\}$. $\mathcal{T}$ has five failure inducing sets, shown in 
the adjacent table,
along with their corresponding variable node and check node sets that do not eventually converge or are not eventually satisfied, respectively. Observe that $\mu(\mathcal{T}) = 3$ and the strength of $\mathcal{T}$ is four. Since each failure inducing set leads to check nodes that are not eventually satisfied, the corresponding type of error is a mismatched syndrome.  \hfill $\Box$ 
\end{exa}

Absorbing sets are a related combinatorial structure that characterize iterative decoder failure of classical LDPC codes in many settings \cite{dolecek07}. In examining syndrome-based iterative decoding, we found that many failure inducing sets corresponded to absorbing sets, motivating the investigation in this paper.

\begin{defi}
An \textit{$(a,b)$-absorbing set} $\mathcal{A}$ in a Tanner graph $\cG$ is a subset of variable nodes such that $|\mathcal{A}| = a$, there are $b$ odd degree vertices in $W_{\mathcal{A}}$, and every variable node  $v\in \mathcal{A}$ has more even degree than odd degree neighbors in $\cG_{\mathcal{A}}$. We refer to $\cG_{\mathcal{A}}$ as the {\em absorbing set graph}.
\end{defi} 

In Figure 1, the set $\mathcal{A} = \{v_1,v_2,v_3,v_4\}$ is a $(4,2)$-absorbing set since there are four variable nodes in the set,  two odd degree check nodes in the graph $\mathcal{G}_{\mathcal{A}}$, and each variable node has more even degree check neighbors than odd degree. Similarly, $\{v_1,v_2,v_3\}$ and $\{v_1,v_3,v_4\}$ are $(3,3)$-absorbing sets. However,  the remaining two failure inducing sets are  not absorbing sets since each has  variable nodes with more odd degree than even degree check neighbors in its corresponding induced subgraph.\\ 
\indent We conclude this section with another example   to illustrate the extent to which absorbing sets correspond to failure inducing sets of a trapping set.
The purpose of this analysis is to better understand the topology of failure inducing sets and see how they relate with absorbing sets. Similar to Example 1, some of the failure inducing sets are absorbing sets whereas some are not.
For all analyses we have used the syndrome-based Gallager-B iterative decoder. 
\begin{exa}\label{exa-nonabs}
Consider Figure \ref{fig:Fig4NV21} from ~\cite{raveendran2021trapping}. Figure \ref{fig:Fig4NV21}(a) shows a subgraph induced by a $(5,3)$-trapping set $\mathcal{T}$. 
First, the only failure inducing set with three or fewer variable nodes is $\{v_2,v_4,v_5\}$. 
\begin{figure}
\centering
\vspace{-.1 in}
\resizebox{0.45\textwidth}{!}{
    \subfigure[Trapping set from~\cite{raveendran2021trapping}.]{
         \begin{tikzpicture}[thick,scale=1]
\node[circle,draw=black,label={above:\small{$v_1$}}] (v1) at (-2,2){};
\node[circle,draw=black,label={above:\small{$v_2$}}] (v2) at (2,2){};
\node[circle,draw=black,label={below:\small{$v_3$}}] (v3) at (2,-2){};
\node[circle,draw=black,label={below:\small{$v_4$}}] (v4) at (-2,-2){};
\node[circle,draw=black,label={above:\small{$v_5$}}] (v5) at (0,0){};

\node[shape=rectangle,draw=black,label={above:\small{$c_1$}}] (c1) at (0,2){};
\node[shape=rectangle,draw=black,label={right:\small{$c_2$}}] (c2) at (2,0){};
\node[shape=rectangle,draw=black,label={below:\small{$c_3$}}] (c3) at (0,-2){};
\node[shape=rectangle,draw=black,label={left:\small{$c_4$}}] (c4) at (-2,0){};
\node[shape=rectangle,draw=black,label={left:\small{$c_5$}}] (c5) at (-1,1){};
\node[shape=rectangle,draw=black,label={left:\small{$c_6$}}] (c6) at (1,-1){};
\node[shape=rectangle,draw=black,label={left:\small{$c_7$}}] (c7) at (-.5,-.5){};
\node[shape=rectangle,draw=black,label={below:\small{$c_8$}}] (c8) at (1.5,1.5){};
\node[shape=rectangle,draw=black,label={right:\small{$c_9$}}] (c9) at (-1.5,-1.5){};

\path [-,thick] (v1) edge node[left] {} (c1); \path [-,thick] (c1) edge node[left] {} (v2);
\path [-,thick] (v1) edge node[left] {} (c4); \path [-,thick] (c4) edge node[left] {} (v4);
\path [-,thick] (v4) edge node[left] {} (c3); \path [-,thick] (c3) edge node[left] {} (v3);
\path [-,thick] (v2) edge node[left] {} (c2); \path [-,thick] (c2) edge node[left] {} (v3);
\path [-,thick] (v1) edge node[left] {} (c5); \path [-,thick] (c5) edge node[left] {} (v5);
\path [-,thick] (v5) edge node[left] {} (c6); \path [-,thick] (c6) edge node[left] {} (v3);
\path [-,thick] (v5) edge node[left] {} (c7);
\path [-,thick] (v2) edge node[left] {} (c8);
\path [-,thick] (v4) edge node[left] {} (c9);
\end{tikzpicture}
         }
         \hfill
         \subfigure[A (4,4)-absorbing sets]{\raisebox{4mm}{
         \begin{tikzpicture}[thick,scale=1]
\node[circle,draw=black,label={}] (v1) at (-2,2){};
\node[circle,draw=black,label={}] (v2) at (2,2){};
\node[circle,draw=black,label={}] (v3) at (2,-2){};
\node[circle,draw=black,label={}] (v4) at (-2,-2){};

\node[shape=rectangle,draw=black,label={}] (c1) at (0,2){};
\node[shape=rectangle,draw=black,label={}] (c2) at (2,0){};
\node[shape=rectangle,draw=black,label={}] (c3) at (0,-2){};
\node[shape=rectangle,draw=black,label={}] (c4) at (-2,0){};
\node[shape=rectangle,draw=black,label={}] (c5) at (-1,1){};
\node[shape=rectangle,draw=black,label={}] (c6) at (1,-1){};
\node[shape=rectangle,draw=black,label={}] (c8) at (1.5,1.5){};
\node[shape=rectangle,draw=black,label={}] (c9) at (-1.5,-1.5){};

\path [-,thick] (v1) edge node[left] {} (c1); \path [-,thick] (c1) edge node[left] {} (v2);
\path [-,thick] (v1) edge node[left] {} (c4); \path [-,thick] (c4) edge node[left] {} (v4);
\path [-,thick] (v4) edge node[left] {} (c3); \path [-,thick] (c3) edge node[left] {} (v3);
\path [-,thick] (v2) edge node[left] {} (c2); \path [-,thick] (c2) edge node[left] {} (v3);
\path [-,thick] (v1) edge node[left] {} (c5); 
\path [-,thick] (c6) edge node[left] {} (v3);
\path [-,thick] (v2) edge node[left] {} (c8);
\path [-,thick] (v4) edge node[left] {} (c9);
\end{tikzpicture}
         }}
        \hfill
         \subfigure[A non-absorbing set subgraph]
         {\raisebox{0mm}{
         \begin{tikzpicture}[thick,scale=1]
\node[circle,draw=black,label={above:\small{$v_1$}}] (v1) at (-2,2){};
\node[circle,draw=black,label={above:\small{$v_2$}}] (v2) at (2,2){};
\node[circle,draw=black,label={below:\small{$v_4$}}] (v4) at (-2,-2){};
\node[circle,draw=black,label={above:\small{$v_5$}}] (v5) at (0,0){};

\node[shape=rectangle,draw=black,label={above:\small{$c_1$}}] (c1) at (0,2){};
\node[shape=rectangle,draw=black,label={right:\small{$c_2$}}] (c2) at (2,0){};
\node[shape=rectangle,draw=black,label={below:\small{$c_3$}}] (c3) at (0,-2){};
\node[shape=rectangle,draw=black,label={left:\small{$c_4$}}] (c4) at (-2,0){};
\node[shape=rectangle,draw=black,label={left:\small{$c_5$}}] (c5) at (-1,1){};
\node[shape=rectangle,draw=black,label={left:\small{$c_6$}}] (c6) at (1,-1){};
\node[shape=rectangle,draw=black,label={left:\small{$c_7$}}] (c7) at (-.5,-.5){};
\node[shape=rectangle,draw=black,label={below:\small{$c_8$}}] (c8) at (1.5,1.5){};
\node[shape=rectangle,draw=black,label={right:\small{$c_9$}}] (c9) at (-1.5,-1.5){};

\path [-,thick] (v1) edge node[left] {} (c1); \path [-,thick] (c1) edge node[left] {} (v2);
\path [-,thick] (v1) edge node[left] {} (c4); \path [-,thick] (c4) edge node[left] {} (v4);
\path [-,thick] (v4) edge node[left] {} (c3); 
\path [-,thick] (v2) edge node[left] {} (c2); 
\path [-,thick] (v1) edge node[left] {} (c5); \path [-,thick] (c5) edge node[left] {} (v5);
\path [-,thick] (v5) edge node[left] {} (c6); 
\path [-,thick] (v5) edge node[left] {} (c7);
\path [-,thick] (v2) edge node[left] {} (c8);
\path [-,thick] (v4) edge node[left] {} (c9);
\end{tikzpicture}
         }
         }
         }
\vspace{-.05 in}\caption{Failure inducing sets of a $(5,3)$-trapping set.\vspace{-.05 in}}
\label{fig:Fig4NV21}
\end{figure}
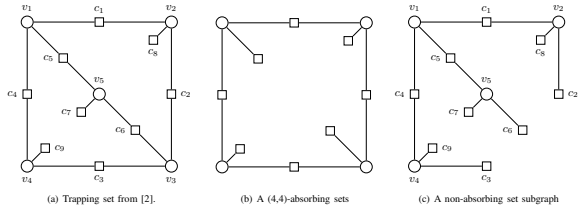
The corresponding input syndrome is $(1,1,1,1,1,1,1,1,1)$, but the decoder oscillates between a syndrome estimate of $(0,0,0,0,0,0,1,1,1)$ and $(0,0,0,0,0,0,0,0,0)$, corresponding to error estimates of all variables in error and none in error, respectively.
The decoder is successful for any other configuration of three or fewer variable nodes in error. \\
\indent All errors of weight four or five result in decoding failure. When four variable nodes are in error there are two types of behaviors shown in Table ~\ref{tab:TSexample2Table}. The first three error patterns correspond to $(4,4)$-absorbing sets whose induced graphs are all isomorphic to Figure~\ref{fig:Fig4NV21}(b). In these cases the decoder converges to a mismatched syndrome. The last two error patterns are not absorbing sets and have induced subgraphs corresponding to Figure~\ref{fig:Fig4NV21}(c). In these cases the decoder oscillates among syndromes, none of which match the input syndrome.\\
\indent Finally, the set of all variable nodes $\{v_1, v_2, v_3, v_4, v_5\}$ forms a $(5,3)$-absorbing set that, when in error, results in decoder failure due to mismatched syndrome. Thus, the entire set of variable nodes is failure inducing. \hfill $\Box$

\begin{table}
\centering
\resizebox{0.48\textwidth}{!}{
\begin{tabular}{c|c|c|c}
    \toprule
    \textbf{Nodes in Error} & \textbf{Input Syndrome} & \textbf{Estimated Syndrome} & \textbf{Estimated Error} \\
    \midrule
    $\{v_1, v_2, v_3, v_4\}$ & $(0,0,0,0,1,1,0,1,1)$ & $(0,0,0,0,0,0,1,0,0)$ & $\{v_5\}$\\ 
    \hline
    $\{v_1, v_2, v_3, v_5\}$ & $(0,0,1,1,0,0,1,1,0)$ & $(0,0,0,0,0,0,0,0,1)$ & $\{v_4\}$\\
    \hline
    $\{v_1,v_3, v_4, v_5\}$ & $(1,1,0,0,0,0,1,0,1)$ & $(0,0,0,0,0,0,0,1,0)$ & $\{v_2\}$\\ 
    \hline
    $\{v_1, v_2, v_4, v_5\}$ & $(0,1,1,0,0,1,1,1,1)$ & $\begin{array}{c}(0,0,0,0,0,0,0,0,0)\\(0,0,0,0,0,0,0,0,0)\\ (1,1,1,1,1,1,0,0,0)\\(1,1,1,1,1,1,0,0,0)\end{array}$ & $\begin{array}{c}\{v_2, v_3, v_4, v_5\}\\ \{ \}\\ \{v_2, v_3, v_4, v_5\}\\ \{v_1, v_3\}\end{array}$\\
    \hline
    $\{v_2, v_3, v_4, v_5\}$ & $(1,0,0,1,1,0,1,1,1)$ & $\begin{array}{c}(1,1,1,1,1,1,0,0,0)\\(1,0,0,1,1,0,1,1,1)\\(0,0,0,0,0,0,0,0,0)\\(0,0,0,0,0,0,0,0,0)\\\end{array}$ & $\begin{array}{c}\{v_1, v_3\}\\  \{v_1,v_2, v_4, v_5\}\\ \{ \}\\ \{v_1, v_2,v_4,v_5\}\end{array}$\\
    \bottomrule
  \end{tabular}
}
\caption{Behavior of weight four error patterns. The first three error patterns correspond to isomorphic (4,4)-absorbing sets in Figure 2(b). The last two error patterns correspond to isomorphic  failure inducing sets in Figure 2(c).}
\vspace{-.1 in}
\label{tab:TSexample2Table}
\end{table}
\end{exa}

\begin{exa}
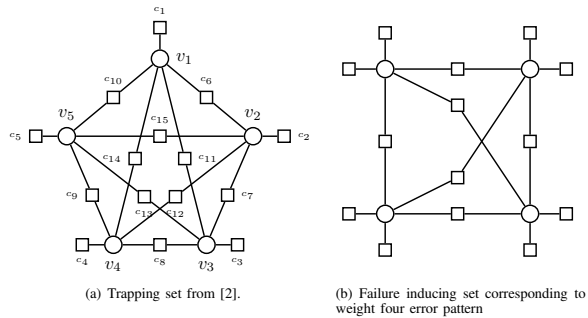
\begin{figure}[!h]
\centering
\resizebox{0.45\textwidth}{!}{
\centering
\subfigure[Trapping set from~\cite{raveendran2021trapping}.]{
         \begin{tikzpicture}[thick,scale=.6]

\node[circle,draw=black,label={right:$v_1$}] (v1) at (0,3){};
\node[circle,draw=black,label={above:$v_2$}] (v2) at (3,.5) {};
\node[circle,draw=black,label={below:$v_3$}] (v3) at (1.5,-3) {};
\node[circle,draw=black,label={below:$v_4$}] (v4) at (-1.5,-3) {};
\node[circle,draw=black,label={above:$v_5$}] (v5) at (-3,.5) {};

\node[shape=rectangle,draw=black,label={above:\tiny$c_1$}] (c1) at (0,4){};
\node[shape=rectangle,draw=black,label={right:\tiny$c_2$}] (c2) at (4,.5){};
\node[shape=rectangle,draw=black,label={below:\tiny$c_3$}] (c3) at (2.5,-3){};
\node[shape=rectangle,draw=black,label={below:\tiny$c_4$}] (c4) at (-2.5,-3){};
\node[shape=rectangle,draw=black,label={left:\tiny$c_5$}] (c5) at (-4,.5){};

\path [-,thick] (v1) edge node[left] {} (c1);
\path [-,thick] (v2) edge node[left] {} (c2);
\path [-,thick] (v3) edge node[left] {} (c3);
\path [-,thick] (v4) edge node[left] {} (c4);
\path [-,thick] (v5) edge node[left] {} (c5);

\node[shape=rectangle,draw=black,label={above:\tiny{$c_6$}}] (c6) at (1.5,1.75){};
\node[shape=rectangle,draw=black,label={above:\tiny$c_{10}$}] (c10) at (-1.5,1.75){};
\node[shape=rectangle,draw=black,label={left:\tiny$c_{9}$}] (c9) at (-2.18,-1.4){};
\node[shape=rectangle,draw=black,label={right:\tiny$c_{7}$}] (c7) at (2.18,-1.4){};
\node[shape=rectangle,draw=black,label={below:\tiny$c_{8}$}] (c8) at (0,-3){};

\node[shape=rectangle,draw=black,label={above:\tiny$c_{15}$}] (c15) at (0,0.5){};
\node[shape=rectangle,draw=black,label={right:\tiny$c_{11}$}] (c11) at (0.7985,-.22){};
\node[shape=rectangle,draw=black,label={left:\tiny$c_{14}$}] (c14) at (-0.7985,-.22){};
\node[shape=rectangle,draw=black,label={below:\tiny$c_{12}$}] (c12) at (0.5,-1.45){};
\node[shape=rectangle,draw=black,label={below:\tiny$c_{13}$}] (c13) at (-0.5,-1.45){};

\path [-,thick] (v1) edge node[left] {} (c6); \path [-,thick] (c6) edge node[left] {} (v2);
\path [-,thick] (v2) edge node[left] {} (c7); \path [-,thick] (v3) edge node[left] {} (c7);
\path [-,thick] (v3) edge node[left] {} (c8); \path [-,thick] (v4) edge node[left] {} (c8);
\path [-,thick] (v4) edge node[left] {} (c9); \path [-,thick] (v5) edge node[left] {} (c9);
\path [-,thick] (v5) edge node[left] {} (c10); \path [-,thick] (v1) edge node[left] {} (c10);

\path [-,thick] (v1) edge node[left] {} (c11); \path [-,thick] (v3) edge node[left] {} (c11);
\path [-,thick] (v1) edge node[left] {} (c14); \path [-,thick] (v4) edge node[left] {} (c14);
\path [-,thick] (v2) edge node[left] {} (c12); \path [-,thick] (v4) edge node[left] {} (c12);
\path [-,thick] (v2) edge node[left] {} (c15); \path [-,thick] (v5) edge node[left] {} (c15);
\path [-,thick] (v5) edge node[left] {} (c13); \path [-,thick] (v3) edge node[left] {} (c13);
\end{tikzpicture}
         }
         \hfill
         \subfigure[Failure inducing set corresponding to weight four error pattern]{\raisebox{4mm}{
         \begin{tikzpicture}[thick,scale=0.7]
\node[circle,draw=black,label={}] (v1) at (-2,2){};
\node[circle,draw=black,label={}] (v2) at (2,2){};
\node[circle,draw=black,label={}] (v3) at (2,-2){};
\node[circle,draw=black,label={}] (v4) at (-2,-2){};

\node[shape=rectangle,draw=black,label={}] (c1) at (0,2){};
\node[shape=rectangle,draw=black,label={}] (c2) at (2,0){};
\node[shape=rectangle,draw=black,label={}] (c3) at (0,-2){};
\node[shape=rectangle,draw=black,label={}] (c4) at (-2,0){};
\node[shape=rectangle,draw=black,label={}] (c5) at (0,1){};
\node[shape=rectangle,draw=black,label={}] (c6) at (0,-1){};
\node[shape=rectangle,draw=black,label={}] (c8) at (3,2){};
\node[shape=rectangle,draw=black,label={}] (c81) at (2,3){}; \path [-,thick] (v2) edge node[left] {} (c81);
\node[shape=rectangle,draw=black,label={}] (c9) at (-3,-2){};
\node[shape=rectangle,draw=black,label={}] (c91) at (-2,-3){};
\path [-,thick] (v4) edge node[left] {} (c91);

\node[shape=rectangle,draw=black,label={}] (c11) at (-2,3){}; \path [-,thick] (v1) edge node[left] {} (c11);
\node[shape=rectangle,draw=black,label={}] (c12) at (-3,2){}; \path [-,thick] (v1) edge node[left] {} (c12);

\node[shape=rectangle,draw=black,label={}] (c31) at (2,-3){}; \path [-,thick] (v3) edge node[left] {} (c31);
\node[shape=rectangle,draw=black,label={}] (c32) at (3,-2){}; \path [-,thick] (v3) edge node[left] {} (c32);

\path [-,thick] (v1) edge node[left] {} (c1); \path [-,thick] (c1) edge node[left] {} (v2);
\path [-,thick] (v1) edge node[left] {} (c4); \path [-,thick] (c4) edge node[left] {} (v4);
\path [-,thick] (v4) edge node[left] {} (c3); \path [-,thick] (c3) edge node[left] {} (v3);
\path [-,thick] (v2) edge node[left] {} (c2); \path [-,thick] (c2) edge node[left] {} (v3);
\path [-,thick] (v1) edge node[left] {} (c5); \path [-,thick] (c5) edge node[left] {} (v3);
\path [-,thick] (c6) edge node[left] {} (v4);
\path [-,thick] (v2) edge node[left] {} (c8);
\path [-,thick] (v4) edge node[left] {} (c9);
\path [-,thick] (c6) edge node[left] {} (v2);
\end{tikzpicture}
         }}
}
\caption{$(5,5)$-trapping set}
\label{fig:Fig7NV21}
\end{figure}

Consider Figure~\ref{fig:Fig7NV21}(a) from ~\cite[Figure~7]{raveendran2021trapping}. First, we found that any error pattern of weight three or fewer is decoded correctly\footnote{In~\cite{raveendran2021trapping}, the authors write that ``for simple binary decoders like syndrome based Gallager-B, any weight three or more error patterns will result in a failure inducing set,'' concluding that the critical number under the Gallager-B decoding is $3$. However, we observed that all weight three errors successfully decoded on the $7^{th}$ iteration.}.

On the other hand, any error pattern of weight four yields an induced subgraph isomorphic to Figure~\ref{fig:Fig7NV21}(b), and is failure inducing.
This configuration is not an absorbing set in that there are variables nodes with an equal number of odd degree and even degree check nodes. Finally, the unique weight five error pattern corresponds to a $(5,5)$-absorbing set and is also a failure inducing set.
\hfill $\Box$
\end{exa}

\section{Absorbing Sets on Their Own}

In this section we aim to understand when absorbing sets are  failure inducing sets with respect to the syndrome-based decoder. We consider absorbing set graphs and analyze which subsets of variable nodes are failure inducing with respect to that graph. We first show that all absorbing sets whose graphs have odd degree check nodes are failure inducing sets with respect to its graph, and therefore are  trapping sets.

\begin{theo}\label{T-oddcheck} 
Let $\cA$ be an $(a,b)$-absorbing set with $b\geq1$.
Then $\cA$ is a failure inducing set. In particular, the decoding syndrome will always be $\vec{0}$ and thus the syndrome value at the odd degree check nodes will never match the input syndrome, thus resulting in a decoding failure.
\end{theo}

\begin{proof}
Let $\cA$ be an absorbing set with at least one odd degree check node. Suppose all variable nodes are in error. Then the input syndrome $\sigma = (\sigma_1,\ldots, \sigma_k)$ has $\sigma_i = 1$ if check node $i$ has odd degree, and $\sigma_i = 0$ otherwise.

The decoder first assumes an all zero error pattern, corresponding to syndrome $\vec{0}$. 
For the next step of the decoding, all even degree check nodes send $0$ and the odd check nodes send $1$. 
Since $\cA$ is an absorbing set, each variable node has strictly more even degree than odd degree check nodes. Thus the error pattern is again $\vec{0}$. When sending information to the check nodes, the extrinsic check nodes are at most evenly tied between odd degree and even degree check nodes. 
If there is a tie, in all cases the variable nodes will send $0$ to every check node. The syndrome is again $\vec{0}$, mismatching the input syndrome at the odd degree check nodes. Additionally, the algorithm is back to the beginning scenario (sending $0$'s to every check node). Therefore the algorithm never terminates, as it decodes to an all zero error pattern and corresponding all zero syndrome at every step.
\end{proof}

We now consider the case when an absorbing set has only even degree check nodes in its graph. While we are primarily concerned with trapping sets in this section, we note a case in which certain error patterns are always failure inducing and hence always form trapping sets. We expand on this idea in Section 5.

For an absorbing set $\cA$, let $H_{\cA}$ be the parity check matrix corresponding to $\cG_{\cA}$.

\begin{theo}\label{all_even_checks} Consider an absorbing set $\cA$ where every check node is of even degree. That is, $\cA$ is an $(a,0)$-absorbing set. For any subset $\cB \subseteq \cA$ whose induced subgraph $\cG_{\cB}$ has no odd degree check nodes, $\cB$ is failure inducing if and only if its corresponding indicator vector is not in the rowspace of $H_{\cA}$. 
\end{theo}

\begin{proof}
Assume all of the variable nodes in $\cB$ are in error and let $\vec{v}$ denote the indicator vector corresponding to $\cB$. Since $\cG_{\cB}$ has no odd degree check nodes, the corresponding input syndrome is $\sigma=\vec{0}$. The decoder assumes an initial error pattern of $\vec{0}$. Then the corresponding estimated syndrome is $\vec{0}$. Since this syndrome matches the input syndrome, the decoding halts and estimates an error pattern $\hat{e}=\vec{0}$. If $\vec{v}$ is a stabilizer, then $e + \hat{e}=\vec{v} + \vec{0}=\vec{v}$ is a stabilizer and the decoder returned a degenerate error. Otherwise there is a logical error and we have decoding failure.
\end{proof}

\begin{cor}\cite[Theorem 2]{10273573}\label{all_even_checks_corollary} Consider an absorbing set $\cA$ where every check node is of even degree. That is, $\cA$ is an $(a,0)$-absorbing set. If $\vec{1}$ is in the rowspace of $H_{\mathcal{A}}$, then $\mathcal{A}$ is not a failure inducing set on $\cG_{\mathcal{A}}$. That is, the decoder will return a degenerate error when all variable nodes are in error. If $\vec{1}$ is not in the rowspace of $H$, then there is a logical error when decoding and $\cA$ is a failure inducing set for $\cG_{\cA}$. 
\end{cor}

\begin{proof}
Assuming all of the variable nodes in $\cA$ are in error, the corresponding input syndrome is $\sigma=\vec{0}$. The decoder assumes an initial error pattern of $\vec{0}$. Then the corresponding estimated syndrome is $\vec{0}$. Since this syndrome matches the input syndrome, the decoding halts and estimates an error pattern $\hat{e}=\vec{0}$. If $\vec{1}$ is a stabilizer, then $e + \hat{e}=\vec{1}$ is a stabilizer and the decoder returned a degenerate error. Otherwise there is a logical error.
\end{proof}

For a special case of Theorem \ref{all_even_checks} we consider a family of graphs known as theta graphs \cite{kelley2008ldpc}.

\begin{exa}\label{theta_graph_example} An $(a,b,c)$-\textit{theta graph}, denoted $T(a,b,c)$, is a graph consisting of two vertices $u$ and $w$, each of degree three, that are connected to each other via three disjoint paths $A$, $B$, $C$ of (edge) lengths $a \geq 1$, $b \geq 1$, and $c \geq 1$, respectively. In Figure \ref{thetagraphs} we see two examples of Tanner graphs forming theta graphs. However, notice that Figure \ref{fig:theta1} forms a $(7,0)$-absorbing set, while Figure \ref{fig:theta2} does not form an absorbing set since all variable node neighbors of $u$ and $w$ have the same number of even and odd degree check neighbors.
When all variable nodes of Figure \ref{fig:theta1} are in error, then Corollary~\ref{all_even_checks_corollary} applies. 
However, we also note that the variable nodes in Figure \ref{fig:theta1} can be partitioned into the union of two smaller $(5,2)$-absorbing sets, namely, the variable nodes in the top and bottom cycles. Indeed, if one of these smaller $(5,2)$-absorbing sets are in error, the decoder fails. 
Similarly, in Figure \ref{fig:theta2}, if either one of the two sets of variable nodes in the bottom and top cycles are in error, there is a decoding failure. Hence these two theta graphs are trapping sets. \hfill $\Box$
\begin{figure}[!h]
\centering
         \resizebox{0.45\textwidth}{!}{
         \subfigure[Case One]{\raisebox{-1mm}{
         \begin{tikzpicture}[thick,scale=.6]
    \draw (0,0) circle (2);
    \draw (0,0) -- (-1,0);
    \draw (-2,0) -- (-1,0);
    \draw (0,0) -- (1,0);
    \draw (1,0) -- (2,0);
    \node[draw,circle,fill=white] at (0:2) {\footnotesize $w$};
    \node[draw,circle,fill=white] at (60:2) {};
    \node[draw,circle,fill=white] at (120:2) {};
    \node[draw,circle,fill=white] at (180:2) {\footnotesize $u$};
    \node[draw,circle,fill=white] at (240:2) {};
    \node[draw,circle,fill=white] at (300:2) {};

    \node[draw,rectangle,fill=white] at (30:2) {};
    \node[draw,rectangle,fill=white] at (90:2) {};
    \node[draw,rectangle,fill=white] at (150:2) {};
    \node[draw,rectangle,fill=white] at (210:2) {};
    \node[draw,rectangle,fill=white] at (270:2) {};
    \node[draw,rectangle,fill=white] at (330:2) {};

    \node[draw,circle,fill=white] at (0:0) {};
    \node[draw,rectangle,fill=white] at (0:-1) {};
    \node[draw,rectangle,fill=white] at (0:1) {};
\end{tikzpicture}
             }\label{fig:theta1}}\quad\quad 
         \hfill
         \subfigure[Case Two]{\raisebox{-1mm}{
         \begin{tikzpicture}[thick,scale=.6]
    \draw (0,0) circle (2);
    \draw (0,0) -- (-1,0);
    \draw (-2,0) -- (-1,0);
    \draw (0,0) -- (1,0);
    \draw (1,0) -- (2,0);
    
    \node[draw,rectangle,fill=white] at (0:2) {\scriptsize $w$};
    \node[draw,rectangle,fill=white] at (60:2) {};
    \node[draw,rectangle,fill=white] at (120:2) {};
    \node[draw,rectangle,fill=white] at (180:2) {\scriptsize $u$};
    \node[draw,rectangle,fill=white] at (240:2) {};
    \node[draw,rectangle,fill=white] at (300:2) {};

    \node[draw,circle,fill=white] at (30:2) {};
    \node[draw,circle,fill=white] at (90:2) {};
    \node[draw,circle,fill=white] at (150:2) {};
    \node[draw,circle,fill=white] at (210:2) {};
    \node[draw,circle,fill=white] at (270:2) {};
    \node[draw,circle,fill=white] at (330:2) {};

    \node[draw,rectangle,fill=white] at (0:0) {};
    \node[draw,circle,fill=white] at (0:-1) {};
    \node[draw,circle,fill=white] at (0:1) {};
\end{tikzpicture}
         }\label{fig:theta2}}
         }
\caption{Non isomorphic $T(6,6,4)$ theta graphs}
\label{thetagraphs}
\end{figure}
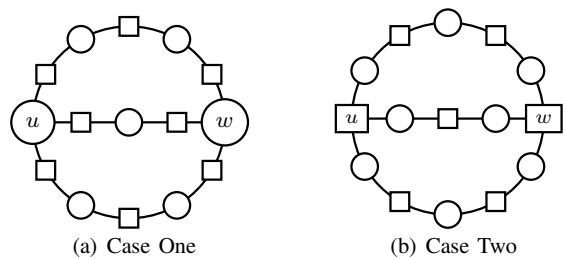
\end{exa}
Although absorbing sets whose graphs have only even degree check nodes are not always failure inducing, we show in the next series of results that such absorbing sets may still be trapping sets due to the presence of proper failure inducing sets.

\begin{lem}
Let $\cA$ be an $(a,0)$-absorbing set such that $\cG_\cA$ is isomorphic to the path $P_{2a-1}$ equal to $v_1c_1v_2c_2\dots c_{a-1}v_a$.
 
Then, $\{v_1\}$ is a failure inducing set. Indeed, on the $\ell^{th}$ iteration the syndrome 
is $\hat{\sigma}^{\ell} = (0,1,1\ldots,1,0,\ldots,0)$ with $\ell -1$ ones, thus it never matches the input syndrome $(1,0,\dots, 0)$. Since the graph is symmetric, a similar statement holds true for $\{v_a\}$.
\end{lem}

\begin{proof}
Set $\ell =0$. The decoder first assumes an all-zero error pattern. 
Therefore 
$\hat{\sigma}^0 = \Vec{0}.$

In the next step of decoding, the check node $c_1$ sends $0 + 1 \equiv 1 \pmod 2$ to $v_1$ and $v_2$. Every other check node sends $0$ to their neighboring variable nodes.

Next, $v_1$ and $v_a$ send $0$ to $c_1$ and $c_{a-1}$, respectively, since they are the endpoints of the path and have no extrinsic check nodes. The variable node $v_2$ sends $0$ to $c_1$ and $1$ to $c_2$. Every other variable node sends $0$ to their neighboring check nodes. Therefore, 
$\hat{\sigma}^{1} = (0,1,0,\dots, 0).$

For $\ell=2$, as in every step, $v_1$ and $v_a$ send 0 to $c_1$ and $c_{a-1}$. But, in this iteration, $v_2$ sends 0 to $c_1$ and 1 to $c_2$, resulting on the syndrome $\hat{\sigma}^2 = (0,1,1,0,\ldots, 0)$.
Continuing in such a fashion we obtain $\hat{\sigma}^{\ell}$ as in the claim. Moreover for $\ell>a$, the syndrome never changes. 

\end{proof}

While it is necessary to identify at least one failure inducing set to characterize whether or not a subset of variable nodes forms a trapping set, we are interested in understanding failure inducing sets more broadly to gain a better understanding of overall decoder performance. Empirical results for paths indicate different phenomena occurring based on the parity of the number of variable nodes $a$. When $a$ is even,  $\vec{1}$ is in the rowspace of the parity check matrix for $\mathcal{G}_{\mathcal{A}}$, and when $a$ is odd, $\vec{1}$ is not in the rowspace of the parity check matrix. Because of this, we observe successful decoding of a large class of errors $e$ when $a$ is even, as the decoder returns $e+\vec{1}$ as the estimated error, successfully returning a degenerate error. However, in the odd case this results in a logical error.

\begin{theo}
    An acyclic $(a,0)$-absorbing set $\cA$ is a trapping set.
\end{theo}
\begin{proof}
Without loss of generality assume $\cG_{\cA}$ is connected (since a union of disconnected absorbing sets is absorbing), so $\cG_{\cA}$ is a tree. Thus, $\cG_\cA$ has at least two leaves, and by assumption, since $\cG_\cA$ has no degree one check nodes, the leaves are variable nodes. Let $v$ be a leaf and $c$ its adjacent check node.
We claim that $\{v\}$ is a failure inducing set. In particular, the input syndrome value at $c$ is 1, but the decoding process always estimates a syndrome value of 0.
This is due to the fact that all incoming messages to $c$ from neighboring variable nodes are 0, resulting in the estimated syndrome value at $c$ being 0, hence mismatched and resulting in decoding failure. The remainder of the proof shows why this is the case. 

First note that degree one variable nodes always send 0 to $c$ since they have no extrinsic variable nodes.
On the other hand, variable nodes of degree strictly higher than two will also send 0 due to majority voting. To see this, see Figure~\ref{acyclic1}.
Decoding starts with all variable nodes sending 0 to their neighbors.
Since the outgoing check node message is the XOR of all extrinsic check nodes and the input syndrome value, $c$ sends 1 to its neighbors and all other check nodes send 0.
Additionally, at variable nodes of degree higher than two (such as $v'$) there are more check nodes sending 0 than sending 1. So there is at most an equal number of extrinsic check nodes sending 0 and sending 1. In event of a tie, 0 is sent. Thus, the incoming messages in $c$ are always 0, and this causes a mismatch.
\begin{figure}[!h]
\centering
\resizebox{0.45\textwidth}{!}{
         \subfigure[\empty]{
         \begin{tikzpicture}[thick,scale=1]
\node[circle,draw=black,label={}] (v1) at (-4,0){};
\node[circle,draw=black,label={[xshift=-0.3cm, yshift=-0.1cm]\Large$v'$}] (v2) at (0,0){};
\node[circle,draw=black,label={}] (v3) at (4,0){};
\node[circle,draw=black,label={}] (v4) at (8,0){};
\node[circle,draw=black,label={}] (v5) at (-2,2){};
\node[circle,draw=black,label={\Large$v$}] (v6) at (2,2){};
\node[circle,draw=black,label={}] (v7) at (2,-2){};
\node[circle,draw=black,label={}] (v8) at (6,-2){};

\node[shape=rectangle,draw=black,label={}] (c1) at (-2,0){\color{red}{0}};
\node[shape=rectangle,draw=black,label={[xshift=-0.25cm, yshift=-0.05cm]\Large$c$}] (c2) at (2,0){\color{red}{1}};
\node[shape=rectangle,draw=black,label={}] (c3) at (6,0){\color{red}{0}};
\node[shape=rectangle,draw=black,label={}] (c4) at (0,2){\color{red}{0}};
\node[shape=rectangle,draw=black,label={}] (c5) at (4,-2){\color{red}{0}};

\path [-,thick] (v1) edge node[left] {} (c1); 
\path [-,thick] (v2) edge node[left] {} (c1);
\path [-,thick] (v2) edge node[left] {} (c2);
\path [-,thick] (v3) edge node[left] {} (c2);
\path [-,thick] (v4) edge node[left] {} (c3);
\path [-,thick] (v3) edge node[left] {} (c3);
\path [-,thick] (v5) edge node[left] {} (c4);
\path [-,thick] (v2) edge node[left] {} (c4);
\path [-,thick] (v6) edge node[left] {} (c2);
\path [-,thick] (v7) edge node[left] {} (c2);
\path [-,thick] (v7) edge node[left] {} (c5);
\path [-,thick] (v8) edge node[left] {} (c5);


\draw[->,blue,thick]  (-3.8,-0.2) to node[right]{\hspace{.1 in}0}(-3,-0.2);
\draw[->,blue,thick]  (-0.2,-0.2) to node[left] {\hspace{-.4 in}0} (-1,-0.2);
\draw[->,blue,thick]  (0.2,0.2) to node[above] {\rotatebox{90}{\hspace{.1 in}0}}(0.2,1);
\draw[->,blue,thick]  (-1.8,1.8) to node[right] {{\hspace{.1 in}0}}(-1,1.8);
\draw[->,blue,thick]  (0.2,-0.2) to node[right] {\hspace{.1 in}0}(1,-0.2);
\draw[->,blue,thick]  (2.2,1.8) to node[below] {\rotatebox{90}{\hspace{-.2 in}0}}(2.2,1);
\draw[->,blue,thick]  (2.2,-1.8) to node[above] {\rotatebox{90}{\hspace{.1 in}0}}(2.2,-1);
\draw[->,blue,thick]  (2.2,-1.8) to node[right] {{\hspace{.1 in}0}}(3,-1.8);
\draw[->,blue,thick]  (5.8,-1.8) to node[left] {{\hspace{-.4 in}0}}(5,-1.8);
\draw[->,blue,thick]  (3.8,-0.2) to node[left] {\hspace{-.4 in}0} (3,-0.2);
\draw[->,blue,thick]  (4.2,-0.2) to node[right] {\hspace{.1 in}0} (5,-0.2);
\draw[->,blue,thick]  (7.8,-0.2) to node[left] {\hspace{-.4 in}0} (7,-0.2);

\end{tikzpicture}
         }\quad\quad 
         \hfill
         \subfigure[\empty]{\raisebox{0mm}{
         \begin{tikzpicture}[thick,scale=1]
\node[circle,draw=black,label={}] (v1) at (-4,0){};
\node[circle,draw=black,label={[xshift=-0.3cm, yshift=-0.1cm]\Large$v'$}] (v2) at (0,0){};
\node[circle,draw=black,label={}] (v3) at (4,0){};
\node[circle,draw=black,label={}] (v4) at (8,0){};
\node[circle,draw=black,label={}] (v5) at (-2,2){};
\node[circle,draw=black,label={\Large$v$}] (v6) at (2,2){};
\node[circle,draw=black,label={}] (v7) at (2,-2){};
\node[circle,draw=black,label={}] (v8) at (6,-2){};

\node[shape=rectangle,draw=black,label={}] (c1) at (-2,0){\color{red}{0}};
\node[shape=rectangle,draw=black,label={[xshift=-0.25cm, yshift=-0.05cm]\Large$c$}] (c2) at (2,0){\color{red}{1}};
\node[shape=rectangle,draw=black,label={}] (c3) at (6,0){\color{red}{0}};
\node[shape=rectangle,draw=black,label={}] (c4) at (0,2){\color{red}{0}};
\node[shape=rectangle,draw=black,label={}] (c5) at (4,-2){\color{red}{0}};

\path [-,thick] (v1) edge node[left] {} (c1); 
\path [-,thick] (v2) edge node[left] {} (c1);
\path [-,thick] (v2) edge node[left] {} (c2);
\path [-,thick] (v3) edge node[left] {} (c2);
\path [-,thick] (v4) edge node[left] {} (c3);
\path [-,thick] (v3) edge node[left] {} (c3);
\path [-,thick] (v5) edge node[left] {} (c4);
\path [-,thick] (v2) edge node[left] {} (c4);
\path [-,thick] (v6) edge node[left] {} (c2);
\path [-,thick] (v7) edge node[left] {} (c2);
\path [-,thick] (v7) edge node[left] {} (c5);
\path [-,thick] (v8) edge node[left] {} (c5);


\draw[->,violet,thick]  (-2.3,-0.2) to node[left]{\hspace{-.35 in}0}(-3,-0.2);
\draw[->,violet,thick]  (-1.7,-0.2) to node[right] {\hspace{.08 in}0} (-1,-0.2);
\draw[->,violet,thick]  (0.2,1.65) to node[below] {\rotatebox{90}{\hspace{-.18 in}0}}(0.2,1.05);
\draw[->,violet,thick]  (-.3,1.8) to node[left] {{\hspace{-.4 in}0}}(-1,1.8);
\draw[->,violet,thick]  (1.7,-0.2) to node[left] {\hspace{-.4 in}1}(1,-0.2);
\draw[->,violet,thick]  (2.2,0.3) to node[above] {\rotatebox{90}{\hspace{.1 in}1}}(2.2,1);
\draw[->,violet,thick]  (2.2,-0.3) to node[below] {\rotatebox{90}{\hspace{-.2 in}1}}(2.2,-1);
\draw[->,violet,thick]  (3.7,-1.8) to node[left] {{\hspace{-.4 in}0}}(3,-1.8);
\draw[->,violet,thick]  (4.3,-1.8) to node[right] {{\hspace{.1 in}0}}(5,-1.8);
\draw[->,violet,thick]  (2.3,-0.2) to node[right] {\hspace{.1 in}1} (3,-0.2);
\draw[->,violet,thick]  (5.7,-0.2) to node[left] {\hspace{-.4 in}0} (5,-0.2);
\draw[->,violet,thick]  (6.3,-0.2) to node[right] {\hspace{.1 in}0} (7,-0.2);

\end{tikzpicture}
         }}
         }
         \resizebox{0.45\textwidth}{!}{
         \subfigure[\empty]{\raisebox{-5mm}{
         \begin{tikzpicture}[thick,scale=1]
\node[circle,draw=black,label={}] (v1) at (-4,0){};
\node[circle,draw=black,label={[xshift=-0.3cm, yshift=-0.1cm]\Large$v'$}] (v2) at (0,0){};
\node[circle,draw=black,label={}] (v3) at (4,0){};
\node[circle,draw=black,label={}] (v4) at (8,0){};
\node[circle,draw=black,label={}] (v5) at (-2,2){};
\node[circle,draw=black,label={\Large$v$}] (v6) at (2,2){};
\node[circle,draw=black,label={}] (v7) at (2,-2){};
\node[circle,draw=black,label={}] (v8) at (6,-2){};

\node[shape=rectangle,draw=black,label={}] (c1) at (-2,0){\color{red}{0}};
\node[shape=rectangle,draw=black,label={[xshift=-0.25cm, yshift=-0.05cm]\Large$c$}] (c2) at (2,0){\color{red}{1}};
\node[shape=rectangle,draw=black,label={}] (c3) at (6,0){\color{red}{0}};
\node[shape=rectangle,draw=black,label={}] (c4) at (0,2){\color{red}{0}};
\node[shape=rectangle,draw=black,label={}] (c5) at (4,-2){\color{red}{0}};

\path [-,thick] (v1) edge node[left] {} (c1); 
\path [-,thick] (v2) edge node[left] {} (c1);
\path [-,thick] (v2) edge node[left] {} (c2);
\path [-,thick] (v3) edge node[left] {} (c2);
\path [-,thick] (v4) edge node[left] {} (c3);
\path [-,thick] (v3) edge node[left] {} (c3);
\path [-,thick] (v5) edge node[left] {} (c4);
\path [-,thick] (v2) edge node[left] {} (c4);
\path [-,thick] (v6) edge node[left] {} (c2);
\path [-,thick] (v7) edge node[left] {} (c2);
\path [-,thick] (v7) edge node[left] {} (c5);
\path [-,thick] (v8) edge node[left] {} (c5);


\draw[->,blue,thick]  (-3.8,-0.2) to node[right]{\hspace{.1 in}0}(-3,-0.2);
\draw[->,blue,thick]  (-0.2,-0.2) to node[left] {\hspace{-.4 in}0} (-1,-0.2);
\draw[->,blue,thick]  (0.2,0.2) to node[above] {\rotatebox{90}{\hspace{.1 in}0}}(0.2,1);
\draw[->,blue,thick]  (-1.8,1.8) to node[right] {{\hspace{.1 in}0}}(-1,1.8);
\draw[->,blue,thick]  (0.2,-0.2) to node[right] {\hspace{.1 in}0}(1,-0.2);
\draw[->,blue,thick]  (2.2,1.8) to node[below] {\rotatebox{90}{\hspace{-.2 in}0}}(2.2,1);
\draw[->,blue,thick]  (2.2,-1.8) to node[above] {\rotatebox{90}{\hspace{.1 in}0}}(2.2,-1);
\draw[->,blue,thick]  (2.2,-1.8) to node[right] {{\hspace{.1 in}1}}(3,-1.8);
\draw[->,blue,thick]  (5.8,-1.8) to node[left] {{\hspace{-.4 in}0}}(5,-1.8);
\draw[->,blue,thick]  (3.8,-0.2) to node[left] {\hspace{-.4 in}0} (3,-0.2);
\draw[->,blue,thick]  (4.2,-0.2) to node[right] {\hspace{.1 in}1} (5,-0.2);
\draw[->,blue,thick]  (7.8,-0.2) to node[left] {\hspace{-.4 in}0} (7,-0.2);

\end{tikzpicture}
             }}\quad\quad 
         \hfill
         \subfigure[\empty]{\raisebox{-5mm}{
         \begin{tikzpicture}[thick,scale=1]
\node[circle,draw=black,label={}] (v1) at (-4,0){};
\node[circle,draw=black,label={[xshift=-0.3cm, yshift=-0.1cm]\Large$v'$}] (v2) at (0,0){};
\node[circle,draw=black,label={}] (v3) at (4,0){};
\node[circle,draw=black,label={}] (v4) at (8,0){};
\node[circle,draw=black,label={}] (v5) at (-2,2){};
\node[circle,draw=black,label={\Large$v$}] (v6) at (2,2){};
\node[circle,draw=black,label={}] (v7) at (2,-2){};
\node[circle,draw=black,label={}] (v8) at (6,-2){};

\node[shape=rectangle,draw=black,label={}] (c1) at (-2,0){\color{red}{0}};
\node[shape=rectangle,draw=black,label={[xshift=-0.25cm, yshift=-0.05cm]\Large$c$}] (c2) at (2,0){\color{red}{1}};
\node[shape=rectangle,draw=black,label={}] (c3) at (6,0){\color{red}{0}};
\node[shape=rectangle,draw=black,label={}] (c4) at (0,2){\color{red}{0}};
\node[shape=rectangle,draw=black,label={}] (c5) at (4,-2){\color{red}{0}};

\path [-,thick] (v1) edge node[left] {} (c1); 
\path [-,thick] (v2) edge node[left] {} (c1);
\path [-,thick] (v2) edge node[left] {} (c2);
\path [-,thick] (v3) edge node[left] {} (c2);
\path [-,thick] (v4) edge node[left] {} (c3);
\path [-,thick] (v3) edge node[left] {} (c3);
\path [-,thick] (v5) edge node[left] {} (c4);
\path [-,thick] (v2) edge node[left] {} (c4);
\path [-,thick] (v6) edge node[left] {} (c2);
\path [-,thick] (v7) edge node[left] {} (c2);
\path [-,thick] (v7) edge node[left] {} (c5);
\path [-,thick] (v8) edge node[left] {} (c5);


\draw[->,violet,thick]  (-2.3,-0.2) to node[left]{\hspace{-.35 in}0}(-3,-0.2);
\draw[->,violet,thick]  (-1.7,-0.2) to node[right] {\hspace{.08 in}0} (-1,-0.2);
\draw[->,violet,thick]  (0.2,1.65) to node[below] {\rotatebox{90}{\hspace{-.18 in}0}}(0.2,1.05);
\draw[->,violet,thick]  (-.3,1.8) to node[left] {{\hspace{-.4 in}0}}(-1,1.8);
\draw[->,violet,thick]  (1.7,-0.2) to node[left] {\hspace{-.4 in}1}(1,-0.2);
\draw[->,violet,thick]  (2.2,0.3) to node[above] {\rotatebox{90}{\hspace{.1 in}1}}(2.2,1);
\draw[->,violet,thick]  (2.2,-0.3) to node[below] {\rotatebox{90}{\hspace{-.2 in}1}}(2.2,-1);
\draw[->,violet,thick]  (3.7,-1.8) to node[left] {{\hspace{-.4 in}0}}(3,-1.8);
\draw[->,violet,thick]  (4.3,-1.8) to node[right] {{\hspace{.1 in}1}}(5,-1.8);
\draw[->,violet,thick]  (2.3,-0.2) to node[right] {\hspace{.1 in}1} (3,-0.2);
\draw[->,violet,thick]  (5.7,-0.2) to node[left] {\hspace{-.4 in}0} (5,-0.2);
\draw[->,violet,thick]  (6.3,-0.2) to node[right] {\hspace{.1 in}1} (7,-0.2);

\end{tikzpicture}
         }}
         }
\caption{Two iterations of the syndrome-based Gallager-B decoder on an acyclic $(a,0)$-absorbing set. The input syndrome is depicted in red, the outgoing  variable node messages in blue, and the outgoing check node messages in cyan.}
\label{acyclic1}
\end{figure}
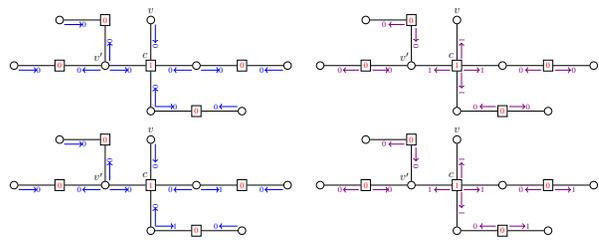
\begin{figure}[!h]
\centering
\resizebox{0.45\textwidth}{!}{
    \subfigure[]{\raisebox{4mm}{
         \begin{tikzpicture}[thick,scale=1]
\node[circle,draw=black,label={}] (v1) at (-2,0){};
\node[circle,draw=black,label={}] (v2) at (0,0){};
\node[circle,draw=black,label={[xshift=-0.2cm, yshift=-0.1cm]$v'$}] (v3) at (2,0){};
\node[circle,draw=black,label={}] (v4) at (4,0){};
\node[circle,draw=black,label={$v$}] (v6) at (1,1){};
\node[circle,draw=black,label={}] (v7) at (1,-1){};
\node[circle,draw=black,label={}] (v8) at (3,-1){};

\node[shape=rectangle,draw=black,label={}] (c1) at (-1,0){};
\node[shape=rectangle,draw=black,label={[xshift=-0.3cm, yshift=-0.1cm]$c$}] (c2) at (1,0){};
\node[shape=rectangle,draw=black,label={}] (c3) at (3,0){};
\node[shape=rectangle,draw=black,label={}] (c5) at (2,-1){};

\path [-,thick] (v1) edge node[left] {} (c1); 
\path [-,thick] (v2) edge node[left] {} (c1);
\path [-,thick] (v2) edge node[left] {} (c2);
\path [-,thick] (v3) edge node[left] {} (c2);
\path [-,thick] (v4) edge node[left] {} (c3);
\path [-,thick] (v3) edge node[left] {} (c3);
\path [-,thick] (v6) edge node[left] {} (c2);
\path [-,thick] (v7) edge node[left] {} (c2);
\path [-,thick] (v7) edge node[left] {} (c5);
\path [-,thick] (v8) edge node[left] {} (c5);


\end{tikzpicture}
         }}
         \hfill \quad
         \subfigure[]{\raisebox{4mm}{
         \begin{tikzpicture}[thick,scale=1]
\node[circle,draw=black,label={$v$}] (v1) at (-2,0){};
\node[circle,draw=black,label={$v'$}] (v2) at (0,0){};
\node[circle,draw=black,label={$v''$}] (v3) at (2,0){};
\node[circle,draw=black,label={}] (v4) at (4,1){};
\node[circle,draw=black,label={}] (v8) at (4,-1){};

\node[shape=rectangle,draw=black,label={$c$}] (c1) at (-1,0){};
\node[shape=rectangle,draw=black,label={}] (c2) at (1,0){};
\node[shape=rectangle,draw=black,label={}] (c3) at (3,1){};
\node[shape=rectangle,draw=black,label={}] (c4) at (3,-1){};

\path [-,thick] (v1) edge node[left] {} (c1); 
\path [-,thick] (v2) edge node[left] {} (c1);
\path [-,thick] (v2) edge node[left] {} (c2);
\path [-,thick] (v3) edge node[left] {} (c2);
\path [-,thick] (v4) edge node[left] {} (c3);
\path [-,thick] (v3) edge node[left] {} (c3);
\path [-,thick] (v8) edge node[left] {} (c4);
\path [-,thick] (v3) edge node[left] {} (c4);
\end{tikzpicture}
         }}
        \hfill \quad\quad
         \subfigure[]
         {\raisebox{4mm}{
         \begin{tikzpicture}[thick,scale=1]
\node[circle,draw=black,label={$v$}] (v1) at (-2,0){};
\node[circle,draw=black,label={$v'$}] (v2) at (0,0){};
\node[circle,draw=black,label={}] (v3) at (2,0){};
\node[circle,draw=black,label={}] (v4) at (4,0){};
\node[circle,draw=black,label={}] (v5) at (-1,1){};
\node[circle,draw=black,label={$v''$}] (v6) at (1,1){};
\node[circle,draw=black,label={}] (v7) at (1,-1){};
\node[circle,draw=black,label={}] (v8) at (3,1){};

\node[shape=rectangle,draw=black,label={$c$}] (c1) at (-1,0){};
\node[shape=rectangle,draw=black,label={[xshift=-0.3cm, yshift=-0.0cm]$c'$}] (c2) at (1,0){};
\node[shape=rectangle,draw=black,label={}] (c3) at (3,0){};
\node[shape=rectangle,draw=black,label={}] (c4) at (0,1){};
\node[shape=rectangle,draw=black,label={}] (c5) at (2,1){};

\path [-,thick] (v1) edge node[left] {} (c1); 
\path [-,thick] (v2) edge node[left] {} (c1);
\path [-,thick] (v2) edge node[left] {} (c2);
\path [-,thick] (v3) edge node[left] {} (c2);
\path [-,thick] (v4) edge node[left] {} (c3);
\path [-,thick] (v3) edge node[left] {} (c3);
\path [-,thick] (v5) edge node[left] {} (c4);
\path [-,thick] (v6) edge node[left] {} (c4);
\path [-,thick] (v6) edge node[left] {} (c2);
\path [-,thick] (v7) edge node[left] {} (c2);
\path [-,thick] (v6) edge node[left] {} (c5);
\path [-,thick] (v8) edge node[left] {} (c5);


\end{tikzpicture}
         }
         }
         }
\vspace{-.05 in}\caption{Graph structures for the remaining subcases. \vspace{-.2 in}}
\label{acyclic2}
\end{figure}
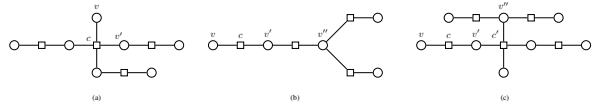

\vspace{0.35cm}

Now consider when $c$ has neighboring variable nodes of degree exactly two (excluding $v$). We consider what happens on any of these branches from $c$ (see Figure~\ref{acyclic2}). The first possibility is when a branch from one of these degree two variable nodes contains all degree two nodes and thus forms a path, as in Figure~\ref{acyclic2}(a). The second is when the closest node to $c$ on the branch that has degree larger than two is a variable node, as in Figure~\ref{acyclic2}(b). The third is when the closest node to $c$ on the branch that has degree larger than two is a check node, as in Figure~\ref{acyclic2}(c). In all these subcases, any degree two check node adjacent to $c$ always receives 0 from their neighbors (extrinsic to $c$). Thus, all incoming messages to $c$ are 0, resulting in a mismatched syndrome.
\end{proof}


\begin{theo}
Consider an $(a,0)$-absorbing set $\cA$ such that $\cG_{\cA}$ forms a cycle. Then $\cG_{\cA}$ is a trapping set. Labelling the graph as in the following diagram, we have the following two cases:

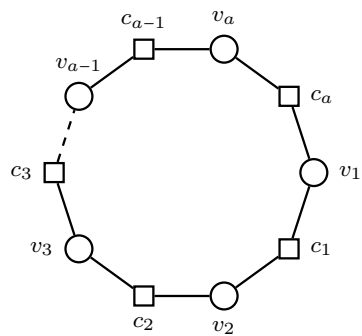
\begin{figure}[!h]
\centering 
\resizebox{0.3\textwidth}{!}{
    \begin{tikzpicture}[thick,scale=0.8]



\foreach \index in {1}
{
    \node[circle, draw=black, label={right: \footnotesize{$v_{\index}$}}] (v1) at ({-((\index-1)*72)}:2) {};
}

\foreach \index in {2}
{
    \node[circle, draw=black, label={below: \footnotesize{$v_{\index}$}}] (v2) at ({-((\index-1)*72)}:2) {};
}

\foreach \index in {3}
{
    \node[circle, draw=black, label={left: \footnotesize{$v_{\index}$}}] (v3) at ({-((\index-1)*72)}:2) {};
}

\node[circle, draw=black, label={above: \footnotesize{$v_{a-1}$}}] (v4) at ({-((4-1)*72)}:2) {};

\node[circle, draw=black, label={above: \footnotesize{$v_{a}$}}] (v5) at ({-((5-1)*72)}:2) {};

\foreach \index in {1}
{
    \node[shape=rectangle, draw=black, label={right: \footnotesize{$c_{\index}$}}] (c\index) at ({-((\index-1)*72+36)}:2) {};
}

\foreach \index in {2}
{
    \node[shape=rectangle, draw=black, label={below: \footnotesize{$c_{\index}$}}] (c\index) at ({-((\index-1)*72+36)}:2) {};
}

\foreach \index in {3}
{
    \node[shape=rectangle, draw=black, label={left: \footnotesize{$c_{\index}$}}] (c\index) at ({-((\index-1)*72+36)}:2) {};
}

\node[shape=rectangle, draw=black, label={above: \footnotesize{$c_{a-1}$}}] (c4) at ({-((4-1)*72+36)}:2) {};


\node[shape=rectangle, draw=black, label={right: \footnotesize{$c_{a}$}}] (c5) at ({-((5-1)*72+36)}:2) {};

\foreach \index in {1, ...,2,3}
{
    \path [-,thick] (v\index) edge node[left] {} (c\index);  
}
\foreach \index in {1, 2}
{
    \path [-,thick] (c\index) edge node[left] {} (v\the\numexpr\index+1\relax);  
}

\foreach \index in {4,5}
{
    \path [-,thick] (v\index) edge node[left] {} (c\index);  
}
\foreach \index in {4}
{
    \path [-,thick] (c\index) edge node[left] {} (v\the\numexpr\index+1\relax);  
}

\foreach \index in {3}
{
    \path [dashed,thick] (c\index) edge node[left] {} (v\the\numexpr\index+1\relax);  
}

\path [-,thick] (c5) edge node[left] {} (v1);  

\end{tikzpicture}
         }
\vspace{-.05 in}\caption{Cycle $\cG_{\cA}$ formed by an $(a,0)$-absorbing set $\cA$}
\label{fig:cycle}
\end{figure}

\begin{center}

    \begin{alphalist}
        \item When $a$ is even, any singleton forms a failure inducing set.\\
        \item When $a$ is odd, $\{v_1, v_2, v_{\frac{a+3}{2}}\}$ forms a failure inducing set.
    \end{alphalist}
\end{center}
\end{theo}

\begin{proof}
    We first consider the case where $a$ is even. Without loss of generality suppose $v_1$ is in error. Then the input syndrome is 
    $\sigma = (1, 0,0, \dots, 0, 1).$
   The decoder starts with all variable nodes sending $0$ to their neighbors. Given the decoding algorithm rules, $c_1$ and $c_a$ send $0$ to their neighbors. All other check nodes send $0$. Since each variable node is of degree $2$, each variable node swaps the message it receives from its neighbors. That is, $v_i$ sends what it receives from $c_{i-1}$ to $c_i$ and vice versa. Because of this, $v_1$ sends $1$ to both its neighbors, $v_2$ sends $0$ to $c_1$ and $1$ to $c_2$, $v_a$ sends $0$ to $c_a$ and $1$ to $c_{a-1}$, and all other variable nodes send $0$. Since all check nodes are also degree $2$ and have syndrome value $0$ except for $c_1$ and $c_a$, this behavior moves around the cycle. That is, at iteration $i$ for $2 \leq i \leq \frac{a+2}{2}$, $v_1$ sends $1$ to its neighbors, $v_i$ sends $0$ to $c_{i-1}$ and $1$ to $c_i$, $v_{a-(i-2)}$ sends $0$ to $c_{a-(i-2)-1}$ and $1$ to $c_{a-(i-2)}$, and all other variable nodes send $0$. For $\frac{a}{2}+1 \leq i \leq a$, $v_1$ sends $1$ to its neighbors, $v_i$ sends $1$ to $c_{i-1}$ and $0$ to $c_i$, $v_{a-(i-2)}$ sends $1$ to $c_{a-(i-2)-1}$ and $0$ to $c_{a-(i-2)}$, and all other variable nodes send $0$. At iteration $a+1$ all variable nodes send $0$, and the decoder is back at the first step of the decoding. Therefore $\sigma^{\ell} \neq \sigma$ for any iteration $\ell$, and hence any singleton is a failure inducing set. 

Now consider the case when $a$ is odd.
The set $\{v_1, v_2, v_{\frac{a+3}{2}}\}$ corresponds to the error vector $e$ where 

$$e_j = \begin{cases}
1, & j=1,2, \frac{a+3}{2},\\
0, & \text{else}.\\
\end{cases}$$

The corresponding input syndrome $\sigma$ is
$$\sigma_j = \begin{cases}
1, & j=2, \frac{a+3}{2}-1, \frac{a+3}{2}, a,\\
0, & \text{else}.\\
\end{cases}$$

We show that on iteration $j=\frac{a+3}{2}$ the estimated syndrome $\hat{\sigma}^{(\frac{a+3}{2})}$ matches the input syndrome $\sigma$ and returns an estimated error $\hat{e}=\vec{0}$. However, $e \oplus \hat{e}=e$ is not in the rowspace of $H$, and hence there is a logical error. That is, there is a decoding failure. 

Since the variable nodes and check nodes are all degree two, at each stage of decoding the variable nodes swap the messages they receive from their neighbors. Similarly, all check nodes with input syndrome value $0$ also swap the messages they receive from their neighbors. Starting with all variable nodes initially sending $0$ to their neighbors, all check nodes correspondingly send $0$ back to their neighbors except for $c_2$, $c_{\frac{a+3}{2}-1}$, $c_{\frac{a+3}{2}}$, and $c_a$. At iteration $2$, $c_2$ sends $1$ to $v_3$ and this $1$ moves around the cycle such that at iteration $\frac{a+3}{2}$, $1$ is sent from $v_{\frac{a+3}{2}-1}$ to $c_{\frac{a+3}{2}-1}$. Likewise, at iteration $\frac{a+3}{2}$, $v_{\frac{a+3}{2}+1}$ sends $1$ to $c_{\frac{a+3}{2}}$. All check nodes besides $c_2$, $c_a$, $c_{\frac{a+3}{2}-1}$, and $c_{\frac{a+3}{2}}$ receive zeroes from both neighbors, so the estimated syndrome at iteration $\frac{a+3}{2}$ matches the input syndrome. When computing the estimated error, $v_1$, $v_2$, $v_{\frac{a+3}{2}-1}$, and $v_{\frac{a+3}{2}+1}$ each receive a $0$ and a $1$ from their neighbors. All other variable nodes receive $0$ from both neighbors. The tie-breaking rule for the Gallager $B$ decoding algorithm results in an estimated error of $\hat{e}=\vec{0}$. However, $e \oplus \hat{e}=e$ is not in the rowspace of $H$, so there is a decoding failure.
\end{proof}

\section{Beyond Symmetric Stabilizers}\label{section-ss}

\subsection{Symmetric Stabilizers} 

It was shown in~\cite{raveendran2021trapping} 
that Tanner graphs with symmetric substructures are particularly detrimental to iterative decoders. This is because the symmetry will force the decoder to oscillate between errors with the same syndrome. 
If said substructure stems from a stabilizer, then it will capture degenerate errors.
This scenario is the counterpart of Theorem~\ref{all_even_checks} which deals with non-stabilizers.
\begin{defi}\cite[Def.~5]{raveendran2021trapping}\label{D-ss}
    A \emph{symmetric stabilizer} is a stabilizer with the set of variable nodes, whose induced subgraph has no odd degree check nodes, and that it can be partitioned into an even  number of disjoint subsets such that 
    \begin{alphalist}
        \item subgraphs induced by these subsets of variable nodes are isomorphic, and,
        \item each subset has the same set of odd degree check nodes neighbors in its induced graph.
    \end{alphalist}
\end{defi}
It follows directly by the definition that a symmetric stabilizer has only even degree check nodes. Thus, the following is immediate.
\begin{theo}\label{T-ss}
    A symmetric stabilizer $S$ is an $(|S|,0)$-absorbing set.
\end{theo}

\begin{rem}
    We have already seen that $(a,0)$-absorbing sets force the decoder to oscillate between errors, which given the work of~\cite{raveendran2021trapping} and the above result, comes as no surprise.
    However, having no odd degree check nodes only partially explains the decoder failure. As we will show in Theorem~\ref{absorbing_union}, the decoder somewhat exclusively depends on the structure of the shared check nodes between the isomorphic constituents. Furthermore, the same theorem shows that the isomorphism condition can be dropped, and the number of constituents can be even or odd; see also Example~\ref{exa-3noniso}.
\end{rem}

\begin{rem}
    The constituents of a symmetric stabilizer may or may not be absorbing sets of their own right (see Figure~\ref{fig:SymmStab} when they are and Figure~\ref{F-10ss} when they are not). 
    If the constituents do form absorbing sets, then they necessarily have $b\geq 1$ odd degree check nodes, and therefore form failure inducing sets by Theorem~\ref{T-oddcheck}.
\end{rem}

\begin{exa}
\begin{figure}[!h]
\centering
\resizebox{0.3\textwidth}{!}{
\begin{tikzpicture}[thick,scale=0.8]


\foreach \index in {1, ..., 9, 10}
{
    \node[circle, draw=black, minimum size=4.5mm] (v\index) at ({-((\index-1)*36)+90}:4.5) {};
}

\foreach \index in {1, ..., 9, 10}
{
    \node[shape=rectangle, draw=black, minimum size=3mm] (c\index) at ({-((\index-1)*36+18)+90}:4.5) {};
}

\foreach \index in {1, ..., 9, 10}
{
    \path [-,thick] (v\index) edge node[left] {} (c\index);  
}
\foreach \index in {1, 2, ..., 9}
{
    \path [-,thick] (c\index) edge node[left] {} (v\the\numexpr\index+1\relax);  
}
\path [-,thick] (c10) edge node[left] {} (v1);  

\foreach \index in {0,1,..., 4}
{
    \node[shape=rectangle, draw=black, minimum size=3mm] (c\the\numexpr\index*2+11\relax) at ({100-72*\index}:3.5) {};
    \node[shape=rectangle, draw=black, minimum size=3mm] (c\the\numexpr\index*2+12\relax) at ({80-72*\index}:3.5) {};
}
\foreach \index in {1,3,5,7,9}
{
    \path [-,thick] (v\index) edge node[left] {} (c\the\numexpr\index+10\relax);
    \path [-,thick] (v\index) edge node[left] {} (c\the\numexpr\index+11\relax);  
}
\path [-,thick] (v8) edge node[left] {} (c11);
\path [-,thick] (v8) edge node[left] {} (c16);
\path [-,thick] (v4) edge node[left] {} (c12);
\path [-,thick] (v4) edge node[left] {} (c17);
\path [-,thick] (v10) edge node[left] {} (c13);
\path [-,thick] (v10) edge node[left] {} (c18);
\path [-,thick] (v6) edge node[left] {} (c14);
\path [-,thick] (v6) edge node[left] {} (c19);
\path [-,thick] (v2) edge node[left] {} (c15);
\path [-,thick] (v2) edge node[left] {} (c20);

\foreach \index in {0,1,..., 4}
{
    \node[shape=rectangle, draw=black, minimum size=3mm] (c\the\numexpr\index+21\relax) at ({90-72*\index}:1.5) {};
}
\path [-,thick] (c21) edge node[left] {} (v1);
\path [-,thick] (c21) edge node[left] {} (v6);
\path [-,thick] (c22) edge node[left] {} (v3);
\path [-,thick] (c22) edge node[left] {} (v8);
\path [-,thick] (c23) edge node[left] {} (v5);
\path [-,thick] (c23) edge node[left] {} (v10);
\path [-,thick] (c24) edge node[left] {} (v2);
\path [-,thick] (c24) edge node[left] {} (v7);
\path [-,thick] (c25) edge node[left] {} (v4);
\path [-,thick] (c25) edge node[left] {} (v9);

\end{tikzpicture}
}
\caption{$(10,0)$ symmetric stabilizer.}
\label{F-10ss}
\end{figure}
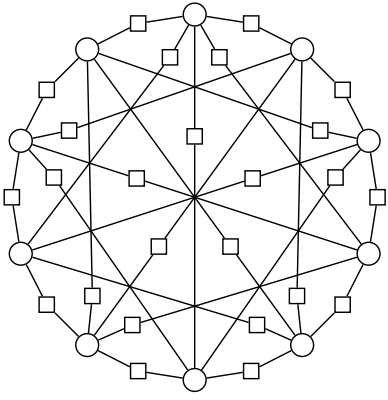

Consider the $(10,0)$-symmetric stabilizer shown in Figure \ref{F-10ss}, first presented in \cite{raveendran2021trapping}. The entire collection of variable nodes forms a $(10,0)$-absorbing set. This graph also contains $\binom{10}{6}= 210$ smaller $(6,12)$-absorbing sets. However, these are the smallest absorbing sets contained within the graph. Thus it is not possible to partition the variable nodes into a disjoint union of absorbing sets. \hfill $\Box$
\end{exa}


For the reminder of this section, we show that the isomorphism condition in Definition~\ref{D-ss} is not necessary and that symmetric stabilizers are themselves a special case of a far more general phenomena. First we recall some necessary notation. Given a Tanner graph $\cG$ and subset $\mathcal{A}\subseteq V(\cG)$ with induced subgraph $\cG_{\mathcal{A}}$, let $\mathcal{O}_{\mathcal{A}}$ denote the set of check nodes in $\cG_{\mathcal{A}}$ with odd degree and $\mathcal{E}_{\mathcal{A}}$ denote the set of check nodes in $\cG_{\mathcal{A}}$ with even degree.

\begin{lem}\label{odd_check_partition}
Given an $(a,0)$-absorbing set $\cA$, if $\cA$ can be partitioned into disjoint absorbing sets $\cA = \cA_1 \cup \cA_2$ for $\cA_1$ an $(a_1, b_1)$-absorbing set and $\cA_2$ an $(a_2, b_2)$-absorbing set, then the odd degree checks of $\cG_{\cA_1}$ and $\cG_{\cA_2}$ are necessarily the same and $b_1=b_2$. 
\end{lem}

\begin{proof}
Consider a check node $c \in \cO_{\cA_1}$. Then 

$$\deg_{\cG_{\cA_1}}(c)=\deg_{\cG_{\cA}}(c)-\deg_{\cG_{\cA_2}}(c).$$

Since $\cA$ is an $(a,0)$-absorbing set, $\deg_{\cG_{\cA}}(c)$ is even and thus $\deg_{\cG_{\cA_2}}(c)$ is odd. Therefore $c \in \cO_{\cA_2}$ and $\cO_{\cA_1} \subseteq \cO_{\cA_2}$. Similarly, $\cO_{\cA_2} \subseteq \cO_{\cA_1}$ so $\cO_{\cA_1}=\cO_{\cA_2}$ and $b_1 =b_2$. 
\end{proof}

\begin{theo}\label{absorbing_union} Consider a Tanner graph $\cG$ containing an $(a,0)$-absorbing set $\mathcal{A}$ such that for every 
$c \in \cN(\cA)$, $\mathcal{N}_{\cG}(c) \subseteq \cA$. Suppose that $\mathcal{A}=\cA_1\cup \cA_2$ where each $\mathcal{A}_i$ is an $(a_i,b)$-absorbing set for $b \ge 1$ and $\mathcal{A}_1 \cap \mathcal{A}_2 = \emptyset$.
Then $\mathcal{A}_1$ (equivalently, $\mathcal{A}_2$) is a failure inducing set.
\end{theo}
\begin{proof}
Without loss of generality suppose $\cA_1$ is in error. Then the input syndrome $\sigma$ has values 

$$\sigma_{j} = \begin{cases}
    1, &\: \: \: \text{for check nodes} \: c_j \in \mathcal{O}_{\cA_1},\\
    0, & \: \: \: \text{else}.
\end{cases}$$

By Lemma \ref{odd_check_partition}, since $\mathcal{A}$ is $(a,0)$-absorbing, the $b$ odd degree check nodes in $\cA_1$ and $\cA_2$ coincide. Thus $\sigma_{j}$ also has value $1$ for the odd degree check nodes $c_j \in \mathcal{O}_{\mathcal{A}_2}$.

Consider $c \in \mathcal{O}_{\cA_1}$. 
In the first iteration of decoding, all variable nodes send $0$ to their neighbors. After this step, $c$ sends $1$ to all its neighbors since $\sigma_c=1$. 

Let $v \in \mathcal{N}(c)$. Then $v \in \cA_1$ or $v \in \cA_2$ but not both. Since both $\cA_1$ and $\cA_2$ are absorbing sets, $v$ has strictly more even degree than odd degree check nodes. 
Thus, for all $c' \in \mathcal{N}(v)$, the set $\mathcal{N}(v)\setminus \{c'\}$ has at most the same number of even degree and odd degree check nodes. 
Equivalently, for all $c' \in \mathcal{N}(v)$, at most the same number of check nodes in $\mathcal{N}(v)\setminus \{c'\}$ send ones as zeros. Hence, in all cases $v$ sends $0$ to all its neighbors. Additionally, all variable nodes not adjacent to check nodes with corresponding input syndrome value 1 send $0$ to their neighbors.
This means that initially $v$ receives zeros from all even degree neighbors and ones from odd degree neighbors. Hence $v$ sends $0$ to all its neighbors. Moreover, $v$ always sends $0$ to all its neighbors since it receives ones only from its odd degree neighbors. 
Thus $v$ always sends $0$ to $c$. Since $v$ was arbitrary, $\hat{\sigma}_c=0$ for all iterations, resulting in decoding failure. Specifically, the estimated syndrome fails to match the input syndrome at precisely the $b$ odd degree check nodes in $\mathcal{O}_{\mathcal{A}_1}$.
\end{proof}

\begin{exa} Figure~\ref{fig:SymmStab} is an example of a symmetric stabilizer whose constituents $\cA_1=\{v_1, v_2, v_5, v_6\}$ and $\cA_2=\{v_3, v_4, v_7, v_8\}$ form $(4,2)$-absorbing sets with isomorphic induced subgraphs. Under the hypotheses of Theorem $\ref{absorbing_union}$, each of these absorbing sets will be failure inducing with respect to $\cG$. However, Theorem $\ref{absorbing_union}$ does not require $\cG_{\cA_1}$ and $\cG_{\cA_2}$ to be isomorphic and thus situates symmetric stabilizers that can be partitioned in this way into a much larger class of harmful graphical structures. \hfill $\Box$
\end{exa}

\begin{figure}[!h]
\centering
\resizebox{0.3\textwidth}{!}{
    \begin{tikzpicture}[scale=.7]
\node[circle,draw=black, scale=0.75, label={above:\footnotesize{$v_1$}}] (v1) at (-3,2){};
\node[circle,draw=black, scale=0.75, label={above:\footnotesize{$v_2$}}] (v2) at (-1,2){};
\node[circle,draw=black, scale=0.75, label={above:\footnotesize{$v_3$}}] (v3) at (1,2){};
\node[circle,draw=black, scale=0.75, label={above:\footnotesize{$v_4$}}] (v4) at (3,2){};
\node[circle,draw=black, scale=0.75, label={below:\footnotesize{$v_5$}}] (v5) at (-3,0){};
\node[circle,draw=black, scale=0.75, label={below:\footnotesize{$v_6$}}] (v6) at (-1,0){};
\node[circle,draw=black, scale=0.75, label={below:\footnotesize{$v_7$}}] (v7) at (1,0){};
\node[circle,draw=black, scale=0.75, label={below:\footnotesize{$v_8$}}] (v8) at (3,0){};





\node[shape=rectangle,draw=black,scale=0.75] (c1) at (-2,2){};
\node[shape=rectangle,draw=black,scale=0.75] (c2) at (0,2){};
\node[shape=rectangle,draw=black,scale=0.75] (c3) at (2,2){};

\node[shape=rectangle,draw=black,scale=0.75] (c4) at (-3,1){};
\node[shape=rectangle,draw=black,scale=0.75] (c5) at (-1,1){};
\node[shape=rectangle,draw=black,scale=0.75] (c6) at (1,1){};
\node[shape=rectangle,draw=black,scale=0.75] (c7) at (3,1){};

\node[shape=rectangle,draw=black,scale=0.75] (c8) at (-2,0){};
\node[shape=rectangle,draw=black,scale=0.75] (c9) at (0,0){};
\node[shape=rectangle,draw=black,scale=0.75] (c10) at (2,0){};

\path [-,thick] (v1) edge node[left] {} (c1); \path [-,thick] (c1) edge node[left] {} (v2);
\path [-,thick] (v2) edge node[left] {} (c2); \path [-,thick] (c2) edge node[left] {} (v3);
\path [-,thick] (v3) edge node[left] {} (c3); \path [-,thick] (c3) edge node[left] {} (v4);

\path [-,thick] (v1) edge node[left] {} (c4); \path [-,thick] (c4) edge node[left] {} (v5);
\path [-,thick] (v2) edge node[left] {} (c5); \path [-,thick] (c5) edge node[left] {} (v6);
\path [-,thick] (v3) edge node[left] {} (c6); \path [-,thick] (c6) edge node[left] {} (v7);
\path [-,thick] (v4) edge node[left] {} (c7); \path [-,thick] (c7) edge node[left] {} (v8);

\path [-,thick] (v5) edge node[left] {} (c8); \path [-,thick] (c8) edge node[left] {} (v6);
\path [-,thick] (v6) edge node[left] {} (c9); \path [-,thick] (c9) edge node[left] {} (v7);
\path [-,thick] (v7) edge node[left] {} (c10); \path [-,thick] (c10) edge node[left] {} (v8);

\end{tikzpicture}
         }
\vspace{-.05 in}\caption{A symmetric stabilizer that can be partitioned into two disjoint $(4,2)$-absorbing sets}
\label{fig:SymmStab}
\end{figure}
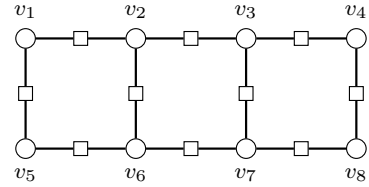

\subsection{Extending the absorbing set partition}

The results in Theorem \ref{absorbing_union} give one instance of when failure inducing sets occur, namely whenever an $(a,0)$-absorbing set embedded in a larger graph can be decomposed into two smaller absorbing sets that intersect precisely at their odd degree check nodes. We present variations of this setting below.

\begin{theo}\label{embedded_union} \textbf{(Union embedded in larger graph, intersecting odd degree checks are subset)}
Consider a Tanner graph $\cG$ containing an $(a_1,b_1)$-absorbing set $\mathcal{A}_1$ and an $(a_2,b_2)$-absorbing set $\mathcal{A}_2$ such that the following conditions hold: 
\begin{alphalist}
    \item $b_1, b_2 \geq 1$
    \item $\mathcal{A}_1 \cap \mathcal{A}_2  = \emptyset$
    \item For every check node $c \in \cE_{\cA_i}, \cN_{\cG}(c)\subseteq \cA_i$.
    \item The intersection $\mathcal{D}:= \mathcal{O}_{\mathcal{A}_1} \cap \mathcal{O}_{\mathcal{A}_2}$ is nonempty.
\end{alphalist}
Under these conditions, the set $\mathcal{A}_1$ (equivalently, $\mathcal{A}_2$) is failure inducing.
\end{theo}

\begin{proof}
    In the first iteration of decoding all variable nodes send 0. In the next iteration, all check nodes $c \in \mathcal{O}_{\cA_1}$ send $1$ and the rest of the check nodes send $0$. Since $\cA_1$ and $\cA_2$ are absorbing sets, the degree condition on the variable nodes results in all nodes $v \in \cA_1 \cup \cA_2$ sending $0$. 
    Next, since $\cN_{\cG}(c)\subseteq \cA_i$ for all $c \in \cE_{\cA_i}$, all check nodes in $\cE_{\cA_1}\cup \cE_{\cA_2}$ receive $0$ and hence send $0$. Thus, we again have the check nodes $c \in \mathcal{O}_{\cA_1}$ sending $1$ and the check nodes in $\cE_{\cA_1}\cup\cE_{\cA_2}$ sending $0$. 
    Hence, in all decoding iterations, the variable nodes in $\cA_1 \cup \cA_2$ send 0 to all their neighbors. In particular, for each check node $c \in \mathcal{D}$ its neighbors $v \in \mathcal{N}_{\cG}(c)$ always send $0$ to $c$. However, $\sigma_c=1$ since $\cA_1$ is in error. Hence $\hat{\sigma}_c \neq \sigma_c$ and we have decoding failure.
\end{proof}
\begin{rem}
Note that in Theorem~\ref{absorbing_union} we require $\mathcal{D}=\mathcal{O}_{\cA_1}=\mathcal{O}_{\cA_2}$. Theorem~\ref{embedded_union} shows that this condition can be relaxed by allowing $\mathcal{D} \subseteq \mathcal{O}_{\cA_1}$ and $\mathcal{D} \subseteq \mathcal{O}_{\cA_2}$, providing this way a further generalization of Definition~\ref{D-ss}.
\end{rem}
\begin{theo}\label{connected_by_single_path} \textbf{(Connected via a single path)} Consider a Tanner graph $\mathcal{G}$ containing a subgraph $\mathcal{H} = \cG_1 \cup \cP \cup \cG_2$ for $V(\cG_1)$ an $(a_1,b_1)$-absorbing set $\mathcal{A}_1$ and $V(\cG_2)$ an $(a_2,b_2)$-absorbing set $\mathcal{A}_2$ with $b_1, b_2 \geq 1$, $\mathcal{A}_1 \cap \mathcal{A}_2 = \emptyset$, and $\cP$ a path connecting a variable node in $\mathcal{A}_1$ to a variable node in $\mathcal{A}_2$. Further suppose for all check nodes $c \in \cE_{\cA_i}$, $\mathcal{N}_{\cG}(c) \subseteq \cA_i$. Then $\mathcal{A}_1$ (equivalently, $\mathcal{A}_2$) is failure inducing.
\end{theo}

\begin{proof}
Without loss of generality suppose $\mathcal{A}_1$ is in error. Let $v \in \cG_1 \cap \cP$, $w \in \cG_2 \cap \cP$, and consider $c \in \mathcal{N}(v) \cap \cP$. Since $\mathcal{A}_1$ is in error, $\sigma_c = 1$. 

The decoding begins with all variable nodes sending $0$. Thus all nodes $c \in \cO_{\cA_1}$ send $1$ to their neighbors and the rest of the check nodes in $\cG$ send $0$. Since $\cA_1$ and $\cA_2$ are absorbing sets, the degree condition on the variable nodes results in all variable nodes in  $\cG_1 \cup \cG_2$ sending $0$. Next, since $\cN_{\cG}(c)\subseteq \cA_i$ for all $c \in \cE_{\cA_i}$, all check nodes in $\cE_{\cA_1}\cup \cE_{\cA_2}$ receive $0$ and hence send $0$. Thus, again we have the check nodes $c \in \mathcal{O}_{\cA_1}$ sending $1$ and the check nodes in $\cE_{\cA_1}\cup\cE_{\cA_2}$ sending $0$. Hence in all decoding iterations, the variable nodes in $\cG_1 \cup \cG_2$ send 0.

Since all nodes in $\cP$ (with the possible exception of $v$ and $w$) are degree 2 and $\sigma_c = 1$, eventually nodes $v$ and $w$ receive $1$ from their check node neighbors in $P$. However, $v \in \cA_1$ and $w \in \cA_2$ again implies that $v$ and $w$ always send $0$ to their neighbors, resulting in $c$ always receiving $0$ from its neighbors. This implies $\hat{\sigma}_c \neq \sigma_c$ and we have decoding failure.
\end{proof}

In the next result we characterize the decoding performance of a specific union of graphs known as dumbbell graphs \cite{kelley2008ldpc}.

\begin{exa}\label{dumbbell}
Define an $(a_1, a_2 ; b)$-\emph{dumbbell graph}, denoted $D(a_1,a_2;b)$ to be a connected graph consisting of two edge-disjoint cycles $A_1$ and $A_2$ of lengths $a_1 \geq 1$ and $a_2 \geq 1$, respectively, that are connected by a path $B$ of length $b \geq 0$.  

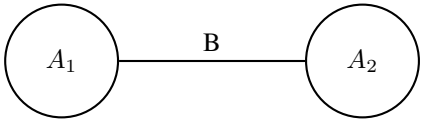
\begin{figure}[!h]
\centering
\begin{tikzpicture}[thick, scale=0.8]
  \node[circle, draw, minimum size=1.5cm] (A1) at (0,0) {$A_1$};
  \node[circle, draw, minimum size=1.5cm] (A2) at (5,0) {$A_2$};
  
  
  \draw[-] (A1) -- (A2) node[midway, above] {B};
\end{tikzpicture}
\caption{Dumbbell graph}
\label{fig-dumbbell}
\end{figure}

Consider a Tanner graph $\cG$ containing a $D(a_1,a_2;b)$ dumbbell graph such that that $A_1$ and $B$ intersect at a variable node and $A_2$ and $B$ intersect at a variable node. Note that the variable nodes of $A_1$ form an $(\frac{a_1}{2},1)$-absorbing set $\cA_1$ and the variable nodes of $A_2$ form an $(\frac{a_2}{2},1)$-absorbing set $\mathcal{A}_2$. Then $\mathcal{A}_1$ (equivalently, $\mathcal{A}_2$) is failure inducing. \hfill $\Box$
\end{exa}

We extend the results of Theorem \ref{connected_by_single_path} by allowing for multiple paths between $\cG_1$ and $\cG_2$. The proof is similar. 

\begin{theo}\label{connected_via_multiple_paths}
\textbf{(Connected via multiple paths)} Consider a Tanner graph $\cG$ containing a subgraph $\cH = \cG_1 \cup \cP(\cG_1,\cG_2) \cup \cG_2$ for $V(\cG_1)$ an $(a_1,b_1)$-absorbing set $\mathcal{A}_1$ and $V(\cG_2)$ an $(a_2,b_2)$-absorbing set $\mathcal{A}_2$ with $b_1, b_2 \geq 1$, $\mathcal{A}_1 \cap \mathcal{A}_2 = \emptyset$, and $\cP(\cG_1, \cG_2)$ a collection of paths between $\cG_1$ and $\cG_2$ intersecting $\cG_1$ and $\cG_2$ at variable nodes. Further suppose for every check node $c \in \cE_{\cA_i}, \cN_{\cG}(c)\subseteq \cA_i$. Then $\cA_1$ (equivalently, $\cA_2$) is failure inducing.
\end{theo} 

\begin{theo}\label{connected_via_a_tree}
\textbf{(Connected via a tree)} Consider a Tanner graph $\cG$ containing a subgraph $\cH = \cG_1 \cup \cT \cup \cG_2$ for $V(\cG_1)$ an $(a_1,b_1)$-absorbing set $\mathcal{A}_1$ and $V(\cG_2)$ an $(a_2,b_2)$-absorbing set $\mathcal{A}_2$ with $b_1, b_2 \geq 1$, $\mathcal{A}_1 \cap \mathcal{A}_2 = \emptyset$, and $\cT$ a tree whose leaves in $\cG_{\cT}$ are elements of $\cA_1$ and $\cA_2$. Further suppose that $|\mathcal{O}_{\cA_1}\cap T|=1$ and for every check node $c \in \cE_{\cA_i}, \cN_{\cG}(c) \subseteq \cA_i$. Then $\cA_1$ is failure inducing. 
\end{theo}

\begin{proof}
Since $\cA_1$ is in error, $\sigma_c =1$ for every $c \in \cO_{\cA_1}$. In particular, $\sigma_c =1$ for every $c \in \cO_{\cA_1}\cap T$.     

The decoding begins with all variable nodes sending $0$. Thus all nodes $c \in \cO_{\cA_1}$ send $1$ to their neighbors and the rest of the check nodes in $\cG$ send $0$. Since $\cA_1$ and $\cA_2$ are absorbing sets, the degree condition on the variable nodes results in all variable nodes in  $A_1 \cup A_2$ sending $0$. Next, since $\cN_{\cG}(c)\subseteq \cA_i$ for all $c \in \cE_{\cA_i}$, all check nodes in $\cE_{\cA_1}\cup \cE_{\cA_2}$ receive $0$ and hence send $0$. Consequently, again we have the check nodes $c \in \mathcal{O}_{\cA_1}$ sending $1$ and the check nodes in $\cE_{\cA_1}\cup\cE_{\cA_2}$ sending $0$. Hence in all decoding iterations, the variable nodes in $\cA_1 \cup \cA_2$ send 0.

Since $|\cO_{\cA_1} \cap T|=1$, all other check nodes in $T$ have input syndrome value $0$. Given that $T$ is acyclic and the variable nodes incident to $T$ and $A_1 \cup A_2$ always send $0$, $c$ always receives $0$ from its neighbors. Thus we have decoding failure.
\end{proof}

\begin{theo}\label{connected_via_multiple_trees}\textbf{(Connected via multiple trees)} Consider a Tanner graph $\cG$ containing a subgraph $\cH = A_1 \cup \cT(A_1,A_2) \cup A_2$ for $V(A_1)$ an $(a_1,b_1)$-absorbing set $\mathcal{A}_1$ and $V(A_2)$ an $(a_2,b_2)$-absorbing set $\mathcal{A}_2$ with $b_1, b_2 \geq 1$, $\mathcal{A}_1 \cap \mathcal{A}_2 = \emptyset$, and $\cT$ a collection of disjoint trees whose leaves in $\cG_{\cT}$ are elements of $\cA_1$ and $\cA_2$. Further suppose that for every tree $\cT \in \cT(\cA_1, \cA_2)$, $|\cO_{\cA_1}\cap \cT|=1$ and for every check node $c \in \cE_{\cA_i}, \cN_{\cG}(c) \subseteq \cA_i$. Then $\cA_1$ is failure inducing. 
\end{theo}

\begin{proof}
Since $\cA_1$ is in error, $\sigma_c =1$ for every $c \in \cO_{\cA_1}$. In particular, $\sigma_c =1$ for every $c \in \cO_{\cA_1}\cap \cT$.     

The decoding begins with all variable nodes sending $0$. Thus all nodes $c \in \cO_{\cA_1}$ send $1$ to their neighbors and the rest of the check nodes in $\cG$ send $0$. Since $\cA_1$ and $\cA_2$ are absorbing sets, the degree condition on the variable nodes results in all variable nodes in  $A_1 \cup A_2$ sending $0$. Next, since $\cN_{\cG}(c)\subseteq \cA_i$ for all $c \in \cE_{\cA_i}$, all check nodes in $\cE_{\cA_1}\cup \cE_{\cA_2}$ receive $0$ and hence send $0$. Consequently, we again have the check nodes $c \in \mathcal{O}_{\cA_1}$ sending $1$ and the check nodes in $\cE_{\cA_1}\cup\cE_{\cA_2}$ sending $0$. Hence in all decoding iterations, the variable nodes in $\cA_1 \cup \cA_2$ send 0.

Since $|\cO_{\cA_1} \cap \cT|=1$ for every $\cT \in \cT(\cA_1,\cA_2)$, all other check nodes in $\cT$ have input syndrome value $0$. Given that $\cT$ is acyclic and the variable nodes incident to $\cT$ and $\cG_1 \cup \cG_2$ always send $0$, each check node $c \in \cO_{\cA_1} \cap \cT$ always receives $0$ from its neighbors. Thus we have decoding failure.
\end{proof}

We further expand the class of harmful graphical substructures by considering partitions of multiple absorbing sets connected by acyclic components, as described in Theorem $\ref{absorbing_partition}$. 

\begin{theo}\label{absorbing_partition}
\textbf{(Partition of many absorbing sets)} 
Consider a Tanner graph $\cG$ containing a subgraph $\cH = \cG_1 \cup \cG_2, \dots, \cG_k \cup \cT(\cG_1,\cG_2, \dots \cup \cG_k)$ for $V(\cG_i)$ an $(a_i,b_i)$-absorbing set $\mathcal{A}_i$ with $b_i \geq 1$ for every $i \in [k]$ and $\cT(\cA_1, \dots, \cA_k)$ a collection of trees in $\cG$ whose leaves in $\cG_{\cT}$ are variable nodes in least two subgraphs in $\{\cG_1, \dots, \cG_k\}$. Further suppose $\cA_i \cap \cA_j = \emptyset$ for all $i, j \in [k]$. Finally suppose that for every tree $\cT \in \cT(\cG_1, \cG_2, \dots, \cG_k)$ incident to $\cA_1$, $|\cO_{\cA_1}\cap \cT|=1$ and for every check node $c \in \cE_{\cA_i}, \cN_{\cG}(c) \subseteq \cA_i$. Then $\cA_1$ is failure inducing. 
\end{theo}
\begin{proof}
Since $\cA_1$ is in error, $\sigma_c =1$ for every $c \in \cO_{\cA_1}$. Without loss of generality, consider $\cT \in  \cT(\cG_1, \cG_2, \dots, \cG_k)$ such that $\cT$ is incident to $\cA_1$ and $\cA_2$. Then in particular, $\sigma_c =1$ for $c \in \cO_{\cA_1}\cap \cT$.      

The decoding begins with all variable nodes sending $0$. Thus, all nodes $c \in \cO_{\cA_1}$ send $1$ to their neighbors and the rest of the check nodes in $\cG$ send $0$. Since $\cA_1$ and $\cA_2$ are absorbing sets, the degree condition on the variable nodes results in all variable nodes in  $\cA_1 \cup \cA_2$ sending $0$. Next, since $\cN_{\cG}(c)\subseteq \cA_i$ for all $c \in \cE_{\cA_i}$, all check nodes in $\cE_{\cA_1}\cup \cE_{\cA_2}$ receive $0$ and hence send $0$. Consequently, again we have the check nodes $c \in \mathcal{O}_{\cA_1}$ sending $1$ and the rest sending $0$. Hence in all decoding iterations, all nodes in $\cA_1 \cup \cA_2$ send 0.

Since $|\cO_{\cA_1} \cap \cT|=1$ for every tree $\cT \in \cT(\cG_1, \cG_2, \dots, \cG_k)$ incident to $\cG_1$, all other check nodes in $\cT$ have input syndrome value $0$. Given that $\cT$ is acyclic and the variable nodes incident to $\cT$ and $\cG_1 \cup \cG_2$ always send $0$, the check node $c \in \cO_{\cA_1} \cap \cT$ always receives $0$ from its neighbors. Thus, we have decoding failure.
\end{proof}

\subsection{Hypergraph-product Codes}
We end the section with an investigation of absorbing sets in a family of QLDPC codes known  as hypergraph-product codes.

\begin{defi}\cite{tillich2013quantum}
Given two classical linear codes $\cC_1$ and $\cC_2$ with corresponding parity check matrices $H_1 \in M_{r_1 \times n_1}(\mathbb{F}_2)$ and $H_2 \in M_{r_2 \times n_2}(\mathbb{F}_2)$, and Tanner graphs $\cG_1$ and $\cG_2$, the corresponding \textit{hypergraph-product code} $\cC'$ is a CSS code with block parity check matrices 

\[
\begin{aligned}
    H_{\sf X} &= (H_1 \otimes I_{n_2} \: \: \: I_{r_1}\otimes H_2^T)\\
    H_{\sf Z} &= (I_{n_1} \otimes H_2 \: \: \: H_1^T \otimes I_{r_2})\\
\end{aligned}
\]
where $I_{r_1}, I_{r_2}, I_{n_1}$, and $I_{n_2}$ are identity matrices of size $r_1, r_2, n_1$, and $n_2$, respectively. 

\end{defi}
\begin{rem}
Given the block structure of the stabilizers of hypergraph-product codes, it is easily seen that the absorbing sets in $G_{\sfX}$ and $G_\sfZ$ are completely characterized by the ``base" matrices $H_1, H_2$ and $H_1^T,H_2^T$ respectively.
Technically speaking, a disconnected union of absorbing sets from the base matrices will give rise to an absorbing set, but these can be discarded since they will not affect one another.
\end{rem}
\begin{exa}\label{exa-3noniso} 
Consider $H_1=H_2=H$ where $H$ is the following parity check matrix for a classical linear code $\cC$.

\[
\scriptsize
H=\begin{pmatrix}
1 & 1 & 0 & 0 & 0 & 0 & 0 & 0 & 0 & 0 & 0 \\
0 & 1 & 1 & 0 & 0 & 0 & 0 & 0 & 0 & 0 & 0 \\
0 & 0 & 1 & 1 & 0 & 0 & 0 & 0 & 0 & 0 & 0 \\
0 & 0 & 0 & 1 & 1 & 0 & 0 & 0 & 0 & 0 & 0 \\
0 & 0 & 0 & 0 & 1 & 1 & 0 & 0 & 0 & 0 & 0 \\
0 & 0 & 0 & 0 & 0 & 1 & 1 & 0 & 0 & 0 & 0 \\
0 & 0 & 0 & 0 & 0 & 0 & 1 & 1 & 0 & 0 & 0 \\
0 & 0 & 0 & 0 & 0 & 0 & 0 & 1 & 1 & 0 & 0 \\
0 & 0 & 0 & 0 & 0 & 0 & 0 & 0 & 1 & 1 & 0 \\
0 & 0 & 0 & 0 & 0 & 0 & 0 & 0 & 0 & 1 & 1 \\
1 & 0 & 0 & 0 & 0 & 0 & 0 & 0 & 0 & 0 & 1 \\
\hline
0 & 1 & 0 & 0 & 0 & 0 & 0 & 0 & 0 & 1 & 0 \\
0 & 0 & 1 & 0 & 0 & 1 & 0 & 0 & 0 & 0 & 0 \\
0 & 0 & 0 & 0 & 0 & 0 & 1 & 0 & 1 & 0 & 0 \\
\end{pmatrix}
\]

The corresponding Tanner graph $\cG$ is shown in Figure \ref{fig:tannergraph}. 

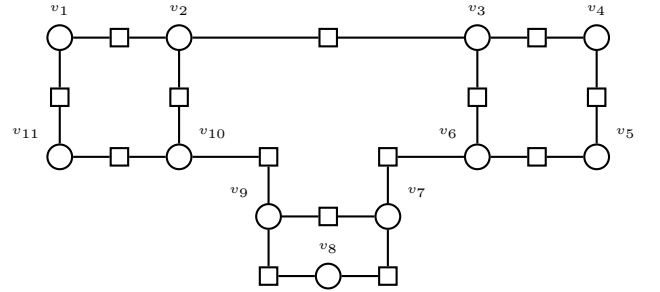
\begin{figure}[!h]
\centering 
\resizebox{0.5\textwidth}{!}{
    \begin{tikzpicture}[thick,scale=.8]

\node[circle,draw=black, label={above:{\tiny{$v_1$}}}] (v1) at (-4.5,3){};
\node[circle,draw=black, label={above:{\tiny$v_2$}}] (v2) at (-2.5,3){};
\node[circle,draw=black, label={above left:{\tiny$v_{11}$}}] (v3) at (-4.5,1){};
\node[circle,draw=black, label={above right:{\tiny$v_{10}$}}] (v4) at (-2.5,1){};

\node[circle,draw=black, label={above:{\tiny$v_3$}}] (v5) at (2.5,3){};
\node[circle,draw=black, label={above:{\tiny$v_4$}}] (v6) at (4.5,3){};
\node[circle,draw=black, label={above left:{\tiny$v_6$}}] (v7) at (2.5,1){};
\node[circle,draw=black, label={above right:{\tiny$v_5$}}] (v8) at (4.5,1){};

\node[circle,draw=black, label={above left:{\tiny$v_9$}}] (v9) at (-1,0){};
\node[circle,draw=black, label={above right:{\tiny$v_7$}}] (v10) at (1,0){};
\node[circle,draw=black, label={above:{\tiny$v_8$}}] (v11) at (0,-1){};

\node[shape=rectangle,draw=black] (c1) at (-3.5,3){};
\node[shape=rectangle,draw=black] (c2) at (-4.5,2){};
\node[shape=rectangle,draw=black] (c3) at (-2.5,2){};
\node[shape=rectangle,draw=black] (c4) at (-3.5,1){};

\node[shape=rectangle,draw=black] (c5) at (3.5,3){};
\node[shape=rectangle,draw=black] (c6) at (2.5,2){};
\node[shape=rectangle,draw=black] (c7) at (4.5,2){};
\node[shape=rectangle,draw=black] (c8) at (3.5,1){};

\node[shape=rectangle,draw=black] (c9) at (0,0){};
\node[shape=rectangle,draw=black] (c10) at (-1,-1){};
\node[shape=rectangle,draw=black] (c11) at (1,-1){};

\node[shape=rectangle,draw=black] (c12) at (0,3){};
\node[shape=rectangle,draw=black] (c13) at (-1,1){};
\node[shape=rectangle,draw=black] (c14) at (1,1){};

\path [-,thick] (v1) edge node[left] {} (c1); \path [-,thick] (c1) edge node[left] {} (v2);
\path [-,thick] (v3) edge node[left] {} (c4); \path [-,thick] (c4) edge node[left] {} (v4);
\path [-,thick] (v1) edge node[left] {} (c2); \path [-,thick] (c2) edge node[left] {} (v3);
\path [-,thick] (v2) edge node[left] {} (c3); \path [-,thick] (c3) edge node[left] {} (v4);

\path [-,thick] (v5) edge node[left] {} (c5); \path [-,thick] (c5) edge node[left] {} (v6);
\path [-,thick] (v7) edge node[left] {} (c8); \path [-,thick] (c8) edge node[left] {} (v8);
\path [-,thick] (v5) edge node[left] {} (c6); \path [-,thick] (c6) edge node[left] {} (v7);
\path [-,thick] (v6) edge node[left] {} (c7); \path [-,thick] (c7) edge node[left] {} (v8);

\path [-,thick] (v9) edge node[left] {} (c9); \path [-,thick] (c9) edge node[left] {} (v10);
\path [-,thick] (v9) edge node[left] {} (c10); \path [-,thick] (c10) edge node[left] {} (v11);
\path [-,thick] (v10) edge node[left] {} (c11); \path [-,thick] (c11) edge node[left] {} (v11);

\path [-,thick] (v2) edge node[left] {} (c12); \path [-,thick] (c12) edge node[left] {} (v5);
\path [-,thick] (v4) edge node[left] {} (c13); \path [-,thick] (c13) edge node[left] {} (v9);
\path [-,thick] (v7) edge node[left] {} (c14); \path [-,thick] (c14) edge node[left] {} (v10);



\end{tikzpicture}
         }
\vspace{-.05 in}\caption{Tanner graph $\cG$ induced by the matrix $H$}
\label{fig:tannergraph}
\end{figure}

All of the variable nodes for $\cG$ form an $(11,0)$-absorbing set. Additionally, $\cG$ contains two $(4,2)$-absorbing sets $\{v_1, v_2, v_{10}, v_{11}\}$ and $\{v_3, v_4, v_5, v_6\}$, and one $(3,2)$-absorbing set $\{v_7, v_8, v_9\}$. Observe that the variable nodes of $\cG$ partition into an odd number of absorbing sets, not all of whom are isomorphic. The decoding performance of these absorbing sets are described by Theorem \ref{all_even_checks} in the case when all variable nodes are in error. These smaller subsets cause decoding failure for $\cC$ as described in Theorem \ref{absorbing_partition}. 

These harmful substructures extend to the resulting hypergraph-product code $\cC'$. The parity check matrices for $\mathcal{C'}$ are 

\[
\begin{aligned}
    H_{\sf X} &= (H \otimes I_{11} \: \: \:  I_{14}\otimes H^T)\\
    H_{\sf Z} &= (I_{11} \otimes H \: \: \: H^T \otimes I_{14}).\\
\end{aligned}
\]

The dimensions for $H_{\sf X}$ and $H_{\sf Z}$ are both $(154,317)$. Since we are ignoring the correlation between $\sf X$ and $\sf Z$ errors, for this example we just consider error patterns for $\sf X$. 

Let $\mathcal{G}_{\sf X}$ be the Tanner graph corresponding to $H_{\sf X}$ and let $\mathcal{A}$ denote the set of all variable nodes in $\mathcal{G}_{\sf X}$. We note that $\mathcal{A}$ contains a $(121,0)$-absorbing set $\cA_1$ resulting from the variable nodes in the $H \otimes I_{11}$ block since all check nodes in this block are of degree two. Therefore, the decoding performance of $\cA_1$ is given by Theorem $\ref{all_even_checks}$. Since $\vec{1}$ is not in the rowspace of $H \otimes I_{11}$, $\cA_1$ causes a logical error in decoding, implying $\cA$ is a trapping set. 

There are eleven isomorphic copies of $\cG$ contained in $\cG_{\sf X}$. Thus, $\cG_{\sf X}$ contains at least eleven $(11,0)$-absorbing sets, twenty-eight $(4,2)$-absorbing sets, and eleven $(3,2)$-absorbing sets which are all guaranteed to cause decoding failure for $\cC'$. In this way we see that our results on the presence of partitions of absorbing sets in the Tanner graph representation of a code imply the presence of harmful weight $3$ and $4$ error patterns for this length $317$ code. 
\hfill $\Box$
\end{exa}

\section{Concluding Observations}

This paper takes a first step towards relating classical absorbing sets  to trapping sets and failure inducing sets of QLDPC codes. Our results show that almost all absorbing set graphs have failure inducing sets and therefore are trapping sets. It remains open to show that any absorbing set will always be a trapping set under this decoder. Our results also demonstrate that many $(a,0)$-absorbing sets form trapping sets when embedded within larger Tanner graphs, especially when they are comprised of certain partitions of smaller absorbing sets. We also situated certain types of symmetric stabilizers, graph structures already identified as being harmful in QLDPC decoding, within a larger class of absorbing set partitions that form failure inducing sets.

Given this identification of the impact of absorbing sets on decoding performance for QLDPC codes, a next step is to understand what can be done to improve code design to mitigate the presence of such harmful graphical substructures. 
For the case of hypergraph-product codes, the $\sfX$-errors and the $\sfZ$-errors are determined both by the base matrices $H_1$, $H_2$ and by the transposed matrices $H_1^T,H_2^T$.
The transposition of course swaps variable nodes and check nodes, and we are currently exploring possible structural connections between the absorbing sets of $\sfX$- and $\sfZ$-errors.
We are also interested in classifying absorbing sets of structured QLDPC codes such as hypergraph-product codes with circulant base matrices.

It would also be interesting to explore how much of this characterization extends to other types of QLDPC decoders, such as Black-Grey-Flip or Belief Propagation. If these characterizations do not extend to these other decoders, we would like to explore what types of graphical characterizations might exist in these cases. Finally, we remark that characterizing failure inducing as in Example~\ref{fig:Fig4NV21}(c) sets that are not absorbing sets remains open.

\bibliographystyle{IEEEtran}
\bibliography{qabs.bib}

\begin{thebibliography}{10}
\providecommand{\url}[1]{#1}
\csname url@samestyle\endcsname
\providecommand{\newblock}{\relax}
\providecommand{\bibinfo}[2]{#2}
\providecommand{\BIBentrySTDinterwordspacing}{\spaceskip=0pt\relax}
\providecommand{\BIBentryALTinterwordstretchfactor}{4}
\providecommand{\BIBentryALTinterwordspacing}{\spaceskip=\fontdimen2\font plus
\BIBentryALTinterwordstretchfactor\fontdimen3\font minus
  \fontdimen4\font\relax}
\providecommand{\BIBforeignlanguage}[2]{{%
\expandafter\ifx\csname l@#1\endcsname\relax
\typeout{** WARNING: IEEEtran.bst: No hyphenation pattern has been}%
\typeout{** loaded for the language `#1'. Using the pattern for}%
\typeout{** the default language instead.}%
\else
\language=\csname l@#1\endcsname
\fi
#2}}
\providecommand{\BIBdecl}{\relax}
\BIBdecl

\bibitem{10273573}
K.~D. Morris, T.~Pllaha, and C.~A. Kelley, ``Analysis of syndrome-based
  iterative decoder failure of qldpc codes,'' in \emph{2023 12th International
  Symposium on Topics in Coding (ISTC)}, 2023, pp. 1--5.

\bibitem{raveendran2021trapping}
N.~Raveendran and B.~Vasi{\'c}, ``Trapping sets of quantum {LDPC} codes,''
  \emph{Quantum}, vol.~5, p. 562, 2021.

\bibitem{Shor-first}
\BIBentryALTinterwordspacing
P.~W. Shor, ``Scheme for reducing decoherence in quantum computer memory,''
  \emph{Phys. Rev. A}, vol.~52, pp. R2493--R2496, Oct 1995. [Online].
  Available: \url{https://link.aps.org/doi/10.1103/PhysRevA.52.R2493}
\BIBentrySTDinterwordspacing

\bibitem{Gottesman-overhead}
D.~Gottesman, ``Fault-tolerant quantum computation with constant overhead,''
  \emph{Quantum Inform. and Computation}, vol.~14, pp. 1338–--1372, Nov 2014.

\bibitem{RBV19}
N.~Raveendran, M.~Bahrami, and B.~Vasic, ``Syndrome-generalized belief
  propagation decoding for quantum memories,'' in \emph{ICC 2019 - 2019 IEEE
  International Conference on Communications (ICC)}, 2019, pp. 1--6.

\bibitem{RRPV22}
N.~Raveendran, N.~Rengaswamy, A.~K. Pradhan, and B.~Vasić, ``Soft syndrome
  decoding of quantum {LDPC} codes for joint correction of data and syndrome
  errors,'' 2022.

\bibitem{R03}
T.~J. Richardson, ``Error floors of {LDPC} codes,'' 2003.

\bibitem{Di02}
C.~Di, D.~Proietti, E.~Telatar, T.~J. Richardson, and R.~L. Urbanke,
  ``Finite-length analysis of low-density parity-check codes on the binary
  erasure channel,'' \emph{IEEE Trans. Inf. Theory}, vol.~48, pp. 1570--1579,
  2002.

\bibitem{dolecek07}
L.~Dolecek, Z.~Zhang, V.~Anantharam, M.~J. Wainwright, and B.~Nikoli{\'c},
  ``Analysis of absorbing sets for array-based {LDPC} codes,'' \emph{2007 IEEE
  International Conference on Communications}, pp. 6261--6268, 2007.

\bibitem{Gottesman-phd97}
D.~Gottesman, ``Stabilizer codes and quantum error correction,'' \emph{PhD
  thesis, California Institute of Technology}, 1997.

\bibitem{CS96}
\BIBentryALTinterwordspacing
A.~R. Calderbank and P.~W. Shor, ``Good quantum error-correcting codes exist,''
  \emph{Phys. Rev. A}, vol.~54, pp. 1098--1105, Aug 1996. [Online]. Available:
  \url{https://link.aps.org/doi/10.1103/PhysRevA.54.1098}
\BIBentrySTDinterwordspacing

\bibitem{Steane96}
\BIBentryALTinterwordspacing
A.~M. Steane, ``Simple quantum error-correcting codes,'' \emph{Phys. Rev. A},
  vol.~54, pp. 4741--4751, Dec 1996. [Online]. Available:
  \url{https://link.aps.org/doi/10.1103/PhysRevA.54.4741}
\BIBentrySTDinterwordspacing

\bibitem{T81}
R.~Tanner, ``A recursive approach to low complexity codes,'' \emph{IEEE
  Transactions on Information Theory}, vol.~27, no.~5, pp. 533--547, 1981.

\bibitem{mackay2004sparse}
D.~J. MacKay, G.~Mitchison, and P.~L. McFadden, ``Sparse-graph codes for
  quantum error correction,'' \emph{IEEE Transactions on Information Theory},
  vol.~50, no.~10, pp. 2315--2330, 2004.

\bibitem{kelley2008ldpc}
C.~A. Kelley and J.~L. Walker, ``Ldpc codes from voltage graphs,'' in
  \emph{2008 IEEE International Symposium on Information Theory}.\hskip 1em
  plus 0.5em minus 0.4em\relax IEEE, 2008, pp. 792--796.

\bibitem{tillich2013quantum}
J.-P. Tillich and G.~Z{\'e}mor, ``Quantum ldpc codes with positive rate and
  minimum distance proportional to the square root of the blocklength,''
  \emph{IEEE Transactions on Information Theory}, vol.~60, no.~2, pp.
  1193--1202, 2013.

\end{thebibliography}
\end{document}